\documentclass[12pt,leqno, oneside]{amsart}

\usepackage{amssymb,amsmath,amsthm, nccmath}
\usepackage{graphicx,color, textcomp}
\usepackage[dvipsnames]{xcolor}
\usepackage{stackrel}
\usepackage{mathtools}
\usepackage[toc,page]{appendix}
\usepackage{comment, mathtools}
\usepackage[T1]{fontenc} 							
\usepackage[utf8]{inputenc} 							
\usepackage[english]{babel} 							

\usepackage[left=2cm,top=2.5cm,bottom=2.5cm,right=2cm]{geometry}
\everymath{\displaystyle}
\usepackage{tikz-cd} 
\usepackage[percent]{overpic}

\def \R {{\mathbb R}}
\def \N {{\mathbb N}}
\def \Z {{\mathbb Z}}

\def \buv {{\bar u, \bar v}}
\def \hi {{\hat i}}
\def \hj {{\hat j}}
\def \hk {{\hat k}}
\def \U  {\mathbb U}

\def \ve {\varepsilon}

\def \eqdef {\doteq}

\def \bu {\bar u}
\def \bv {\bar v}

\def \bp {\bar p}
\def \H {\mathcal H}
\def \S {\mathcal S}
\def \ku {\kappa_1}
\def \kd {\kappa_2}

\newcommand {\beqn} {\begin{equation*}}
	\newcommand {\eeqn}	{\end{equation*}}
\newcommand {\beq} {\begin{equation}}
	\newcommand {\eeq}	{\end{equation}}

\def \ve {\varepsilon}

\newcommand{\scp}[1]{{\langle #1\rangle}}
\newcommand{\un}[1]{\underline #1}
\newcommand{\x}[1]{x^{(#1)}}
\newcommand{\G}[1]{\Gamma^{(#1)}}
\newcommand{\uu}[1]{{\bar u}^{(#1)}}
\newcommand{\vv}[1]{{\bar v}^{(#1)}}

\newtheorem{theorem}{Theorem}[section]
\newtheorem{lemma}[theorem]{Lemma}

\newtheorem{definition}[theorem]{Definition}
\newtheorem{proposition}[theorem]{Proposition}

\newtheorem{remark}[theorem]{Remark}

\numberwithin{equation}{section}

\title[Keplerian billiards in three dimensions]{Keplerian billiards in three dimensions: stability of equilibrium orbits and conditions for chaos}
\author{Irene De Blasi}

\address{University of Turin, Department of Mathematics, via Carlo Alberto 10, Turin, Italy}

\thanks{\noindent Work partially supported by INdAM groups G.N.A.M.P.A.\\
	Supported by Italian Research Center on High Performance Computing Big Data and Quantum
	Computing (ICSC), project funded by European Union - NextGenerationEU - and National
	Recovery and Resilience Plan (NRRP) - Mission 4 Component 2 within the activities of Spoke 3
	(Astrophysics and Cosmos Observations)}

\keywords{billiards, refraction,reflection, Kepler problem, symbolic dynamics,
	topological chaos, linear stability}

\makeatletter
\@namedef{subjclassname@2020}{\textup{2020} Mathematics Subject Classification}
\makeatother

\subjclass[2020] {
		34C28, 
		70F15, 
		37C83, 
		70F16, 
}

\begin{document}
	
	\maketitle
	\baselineskip=18pt	
	\begin{abstract}
		This work presents some results regarding three-dimensional billiards having a non-constant potential of Keplerian type inside a regular domain $D\subset \R^3$. Two models will be analysed: in the first one, only an inner Keplerian potential is present, and every time the particle encounters the boundary of $D$ is reflected back by keeping constant its tangential component to $\partial D$. The second model is a refractive billiard, where the inner Keplerian potential is coupled with a harmonic outer one; in this case, the interaction with $\partial D$ results in a generalised refraction Snell's law. In both cases, the analysis of a particular type of straight equilibrium trajectories, called \emph{homothetic}, is carried on, and their presence is linked to the topological chaoticity of the dynamics for large inner energies.
	\end{abstract}
	\section{Introduction}\label{sec:intro}	
	Three-dimensional billiards, especially in their classical Birkhoff formulation (a ball moving freely inside a convex and closed domain of $\R^3$), have inspired a vast literature since the end of XX century (see \cite{Tabbook} for an extensive summary on the subject). The main problems that have been studied are the presence of periodic trajectories (see for example \cite{babenko1992periodic, farber2002periodic}) and the existence of caustics (as in \cite{berger1995seules}), along with the analysis of equilibrium trajectories (\cite{wojtkowski1990linearly}) and the possible arising of chaotic behaviour (\cite{papenbrock2000numerical,bunimovich1997nowhere}). Besides classical models, some generalisations have also been considered, as in the case of dual billiards (see \cite{tabachnikov2003three}) or more complex systems, like quantum Sinai models (\cite{primack2000quantum}). Moreover, an interesting analogy between outer billiards and planetary motion has been highlighted by Moser (\cite{moser2001stable,moser1978solar}).  \\
	In general dynamical systems, the passage to higher dimensions is usually associated to the arising of a richer dynamics, sometimes not present otherwise: in the case of billiards, it implies that the first return map associated to the system passes from being two-dimensional to act on a four-dimensional phase space. In principle, this could lead to new interesting phenomena, as for example Arnold diffusion (\cite{arnold2014stability}), usually prevented in two dimensions, for example by the presence of invariant curves, such as KAM tori (\cite{moser1962invariant}). \\
	The aim of this work is to introduce the study of three dimensional billiards when the dynamics inside a domain, called $D$, is determined by a Keplerian attractive potential, generated by a mass internal to $D$ itself. To do this, we will consider two different models, which differ from each others by the interaction of the particle with $\partial D$. \\
	\begin{figure}
		\centering
		\includegraphics[width=.2\linewidth]{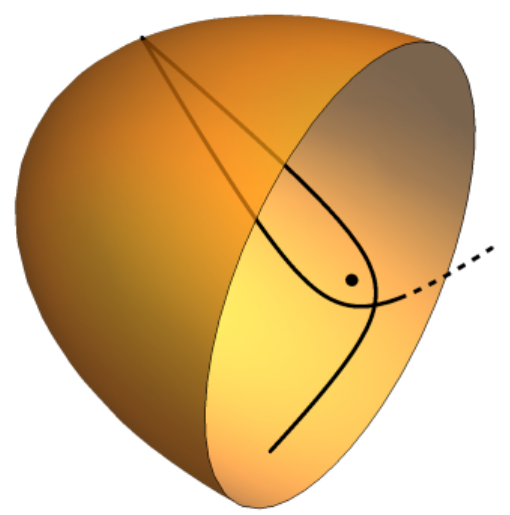}\qquad\qquad\qquad 
		\includegraphics[width=.2\linewidth]{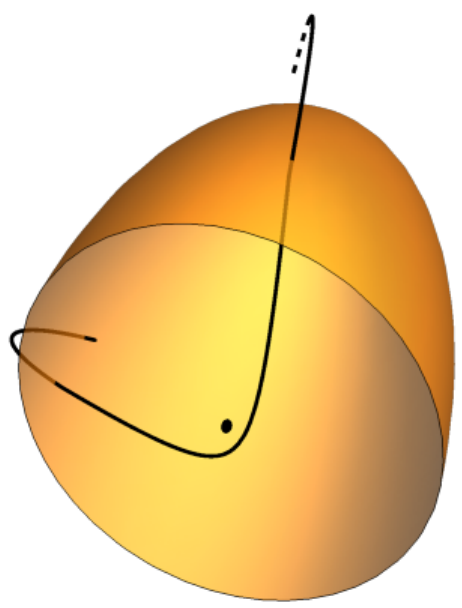}
		\caption{Reflective and refractive billiards in three dimensions; for the reader's convenience, only one half of the boundary $\partial D$ is drawn. Left: reflective case; the billiard trajectory is composed only by inner Keplerian arcs, which are reflected on the boundary. Right: refractive case. The particle enters and exits alternatively from the domain $D$, in a concatenation of outer harmonic and inner Keplerian arcs, connected by a refraction rule which deflects the velocity at every passage trough $\partial D$. }
		\label{fig:modelli_intro}
	\end{figure}
	In the first case, called, in accordance to the known literature, \textit{Kepler billiard}, whenever the test particle encounters the boundary of $D$ it is reflected back by an elastic bounce that keeps constant its velocity's tangent component (see Figure \ref{fig:modelli_intro}, left). The physical interest of such model resides in its possible applications to Celestial Mechanics, as well as in the construction of optical devices. It is for example the case of the motion of a Celestial body which, while attracted by a central mass, is confined inside a region in the space bounded  by a dense outer ring of particles or dust, acting as a wall. A scattering billiard between narrow discs around planetary bodies is proposed by Benet (see \cite{merlo2007strands}), who highlights the strong relation between the model and the Restricted Three Body Problem (\cite{benet1996}), setting up the basis for a study in higher dimensions as well (\cite{benet2005}). Already in 1977, Hitzl and Hénon (\cite{hitzl1977critical}) pointed out the relation between this class of billiards and the Restricted Three Body Problem: they linked the presence of second species periodic solutions,  when the central mass tends to zero, to critical points corresponding to consecutively collisional orbits, formed by Keplerian arcs, in a narrow disc around the central planet. As for the optical applications, it is worth mentioning the behaviour of ultracold neutrons inside the so-called \textit{neutron microscope} (\cite{steyerl1988neutron,masalovich2014remarkable}), which has the property of being reflected against suitable walls, while subjected to gravitational forces. \\
	In the second model presented, the particle is no longer confined inside the domain, and is free to exit undergoing a \textit{refraction} governed by a generalization of the Snell's law. Once outside, the particle is pushed back by a harmonic restoring force, leading to an infinite concatenation of inner Keplerian arcs and outer harmonic ones (see Figure \ref{fig:modelli_intro}, right). We will refer to this second model as the \textit{refraction galactic billiard}, due to its connection with Celestial Mechanics models.  In particular, it is inspired by  \cite{Delis20152448}, where the motion of a particle inside an elliptic galaxy having a central mass is taken into consideration: in this case, the billiard can be interpreted as the \textit{region of influence} of the central body, outside which the gravitational attraction of the whole elliptical distribution of matter prevails. Let us highlight that, apart from the current application to galactic motion, refraction billiards can be useful to model a quite vast class of physical problems, whenever the particle moves under the influence of local forces or different regimes, depending on its position, are present. In particular, an interesting application of refraction rules for random combinatorial billiards can be found in \cite{adams2024toric, defant2024random}.  \\
	In two dimensions, billiard models with non-constant potentials have been the subject of many different papers (\cite{fedorov2001ellipsoidal, pustovoitov2021topological,kobtsev2020elliptic, lerman2021whispering}), investigating for example the possible integrability of the model (\cite{dragovic1997integrable}); more generally, models that also involve an outer dynamics, extending the notion of interaction with the billiard's boundary from a simple reflection, are considered in \cite{gasiorek2021dynamics}. Within the class of gravitational Kepler billiards, it is worth mentioning the works \cite{takeuchi2021conformal, jaud2024geometric}, which prove an interesting result about integrability inside focused conics, and \cite{Bol2017}, where a wider class of billiards with singularities is taken into account; in \cite{takeuchi2023kustaanheimo}, the authors propose a first analysis of the integrability of the three dimensional reflective case by using Kustanheimo-Stiefel regularisation. \\
	In the planar case, some properties of Kepler and galactic refraction billiards have been already depicted in \cite{deblasiterracinirefraction, deblasiterraciniOnSome, deblasiterracinibarutelloChaotic, DebCherBar}: this work is intended as an extension of them, investigating, in particular, the local stability of suitable equilibrium trajectories (see Figure \ref{fig:homothetics}), and, starting from them, the possible presence of chaotic regimes. It is important to highlight how increasing the dimensionality does not lead to a straightforward generalization of what is done on the plane, as the models presented here display a richer dynamics and technical difficulties in the proving schemes arise. First of all, it is worth to ask whether the billiard dynamics is \textit{truly} three-dimensional, since, by conservation of the angular momentum, every arc of our concatenations is planar. As we will see analytically, except in the trivial spherical case both reflection and refraction could change the orbit's plane, then no invariant planes are generically present. Moreover, the interaction of Keplerian and harmonic arcs with a closed surface, instead of a curve, is more complex, and a careful choice of the parameterisation of $\partial D$ is in order. Another important difference regards the treatment of Keplerian collisional trajectories, of paramount importance in the forthcoming analysis. While in two dimensions they can be regularised by following the approach of Levi-Civita (see \cite{maclaurin1742treatise, Levi-Civita}), in higher dimensions another strategy, relying on the construction of a suitable four dimensional space, is necessary; the regularisation we use, due to Kustanheimo and Stiefel, is here presented using Waldvogel's formalism (see \cite{waldvogel2008quaternions}).  We highlight that, although physically not relevant, collision solutions obtained by regularisation methods are of particular importance from an analytical point of view, since this approach allows  to investigate non-collisional solutions close to the collisional ones. This is exactly the type of analysis needed in both Section \ref{sec:stability} and \ref{sec:chaos}, where close-to-collision arcs are considered. \\

	In order to illustrate the principal results of this work, let us give some notation, providing an analytical description of the two models proposed; for this reason, let us take a smooth domain $D\subset \R^3$, of class at least $C^3$.  \\
	In Keplerian billiards, the particle moves inside $\overline D$ along Keplerian geodesics generated by the potential 
		\begin{equation}\label{eq:inner_pot}
			V_I: \overline D\setminus\{0\}\to \R, \quad V_I(z)=\H+\frac{\mu}{|z|}, 
		\end{equation}
		corresponding to the gravitational attraction of a m
		ass with parameter $\mu>0$ situated at the origin. More precisely, any solution of the differential problem 
		\begin{equation}\label{eq:inner_prob_intro}
			\begin{cases}
				z''(t)=\nabla V_I(z(t))\\
				\dfrac12|z'(t)|^2-V_I(z(t))=0
			\end{cases}
		\end{equation}
		with $\H>0$, corresponds to a Keplerian hyperbolic arc with energy equal to $\H$.  Every time our test particle hits the boundary, it is reflected back inside the domain, and continues to move. \\
		In the refractive case the inner Keplerian potential \eqref{eq:inner_pot} is coupled with a harmonic outer one, defined by 
		\begin{equation}\label{eq:outer_pot}
			V_E: \R^3\setminus D\to \R, \quad V_E(z)=\mathcal E -\dfrac{\omega^2}{2}|z|^2,  
		\end{equation}
		with $\mathcal E, \omega>0$. 
		The resulting trajectories are then concatenations of outer elliptic arcs with inner Keplerian ones, joined together on $\partial D$ by a refractive Snell's law which maintains constant the tangent component of the velocity through the transition. \\
	In the present paper we will address, for both models, two main problems. The first one is the existence and stability analysis of a particular class of equilibrium trajectories, called \textit{homothetic} (see again Figure \ref{fig:homothetics}):  these are collision-ejection solutions obtained, via regularization arguments (see  Appendix \ref{sec:app}), along special  directions defined by the so-called \textit{central configurations,} namely, radial directions orthogonal to $\partial D$. More precisely, we can give the following definition. 
		\begin{figure}
		\centering
		\includegraphics[width=0.3\linewidth]{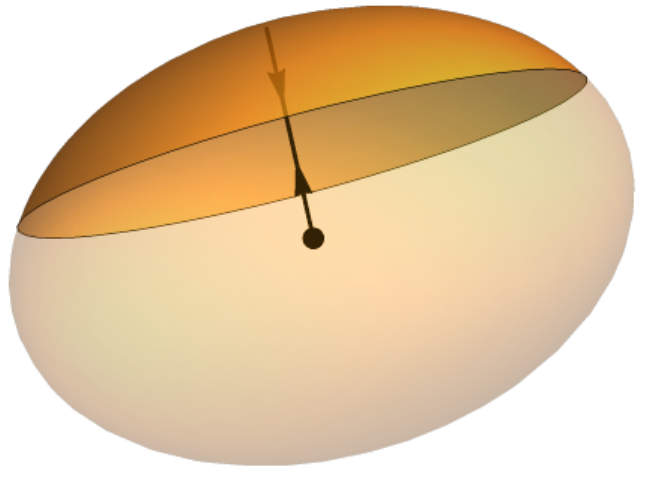}\qquad\qquad\qquad \includegraphics[width=0.3\linewidth]{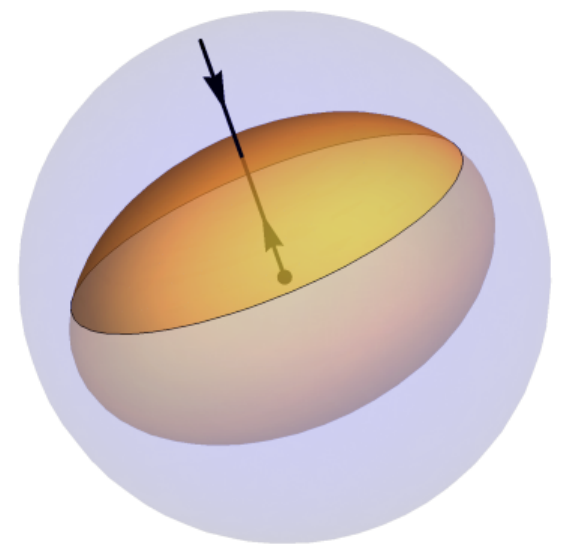}
		\caption{Homothetic collision trajectories in the reflective (left) and refractive (right) case. The particle follows radial arcs in the direction of a central configuration: for the case on the left, the reflection translates into a simple inversion of the trajectory onto itself; for the case on the right, Snell's law does not deflect the velocity vector, since it is orthogonal to the boundary $\partial D$. The blue sphere represents the boundary of the Hill's region associated to the outer potential.   }
		\label{fig:homothetics}
	\end{figure}
	\begin{definition}
		Let us consider $D\subset \R^3$ a bounded and regular domain of $\R^3$, containing the origin, and suppose that it is star-shaped with respect to $0$. A \emph{central configuration} is a point $\bp\in \partial D$ orthogonal to $T_{\bp}\partial D$. 
	\end{definition}
	The existence of central configurations of $D$ is of fundamental importance for the results presented thorough the whole work: the first one regards the stability of homothetic trajectories.  	
	\begin{theorem}\label{thm:stability_intro}
		Let $S\eqdef \partial D$ be the boundary of a three-dimensional billiard, and let $\bp\in S$ a central configuration, $\ku$ and $\kd$ the principal curvatures of $S$ in $\bp$. Then: 
		\begin{enumerate}
			\item in the reflective case, there exists an equilibrium homothetic trajectory of the billiard along the direction of $\bp$, given by a collision-ejection arc, whose stability is given by the quantities 
				\begin{equation}
				\begin{aligned}
					\Delta_1=\frac{32 |\bar p|^4}{\mu^2 \sqrt{VI(\bar p)}}\left(\frac{1}{|\bar p|}-\kappa_1\right)\left[\sqrt{V_I(\bar p)}\left(\frac{1}{|\bar p|}-\kappa_1\right)-\frac{\mu}{2|\bar p|^2\sqrt{V_I(\bar p)}}\right]\\
					\Delta_2=\frac{32 |\bar p|^4}{\mu^2 \sqrt{VI(\bar p)}}\left(\frac{1}{|\bar p|}-\kappa_2\right)\left[\sqrt{V_I(\bar p)}\left(\frac{1}{|\bar p|}-\kappa_2\right)-\frac{\mu}{2|\bar p|^2\sqrt{V_I(\bar p)}}\right];
				\end{aligned}
			\end{equation}
		\item in the refractive case, there is an equilibrium trajectory, composed by the concatenation of an outer homothetic and an inner collision-ejection arc along the direction of $\bar p$, whose stability depends on the quantities, for $i=1,2$, 
		\begin{equation}
			\begin{aligned}
				&\Delta_i=ABCD\\
				& A=\frac{16|\bar p|}{\mathcal E^2\mu^2}\left(\sqrt{V_I(\bp)}\right)\left(\kappa_i-\frac{1}{|\bar p|}\right)\\
				& B= \mathcal E+(\kappa_1|\bar p|-1)\sqrt{V_E(\bar p)}\left(\sqrt{V_E(\bar p )}-\sqrt{V_I(\bar p)}\right)\\
				& C=-\mu\sqrt{V_I(\bar p)}+2|\bar p|\left[\mathcal E+(\kappa_i|\bar p|-1)\sqrt{V_E(\bar p)}\left(\sqrt{V_E(\bar p)}-\sqrt{V_I(\bar p)}\right)\right]\sqrt{V_I(\bar p)}\\
				& D=\mu-2|\bar p|(\kappa_1|\bar p|-1)\left(\sqrt{V_E(\bar p)}-\sqrt{V_I(\bar p)}\right)\sqrt{V_I(\bar{p})}. 
			\end{aligned}
		\end{equation}
		\end{enumerate} 
	In particular, if either $\Delta_1$ or $\Delta_2$ are positive, then the trajectory is linearly unstable, while when both the quantities are negative it is linearly stable. More precisely, in both cases the homothetic trajectory corresponds to a fixed point of a suitable four-dimensional first return map, for which it holds that
	\begin{itemize}
		\item if $\Delta_1\cdot \Delta_2<0$, then we have a saddle-centre fixed point; 
		\item if both $\Delta_1$ and $\Delta_2$ are positive, then the fixed point is of saddle-saddle type; 
		\item if $\Delta_1,\Delta_2<0$, then the fixed point is of centre-centre type. 
	\end{itemize} 
\end{theorem}
The above result is achieved by constructing suitable discrete first return maps, which keeps track of every time the reflective or refractive trajectories hit the boundary $S$, and treating the homothetic equilibrium ones as fixed points. In this framework, it is possible to compute the linear stability of such equilibrium points with the standard tools of discrete dynamical systems (see \cite{devaneyhirschsmale}). \\
According to Theorem \ref{thm:stability_intro}, the linear stability of homothetic equilibrium trajectories depends on the local geometry of $\partial D$ in $\bp$, along with the physical parameters defining the potential $V_I$ and, in the refractive case, $V_E$. From a dynamical point of view, this translates into in the presence of possible \emph{bifurcation phenomena} when suitably changing the values $\H, \mu, \mathcal E, \omega$. \\
The second proposed result regards the arising of chaotic behaviour, provided some simple geometric conditions of the boundary are satisfied: in particular, we will link the presence of at least two central configurations for $D$ which are non antipodal (i.e. the origin does not belong to the segment between them) and non degenerate (namely, they are non degenerate critical points for the function $p\in\partial D\mapsto |p|$) to the existence of a topological chaotic subsystem for sufficiently large inner energies $\H$.
	\begin{theorem}\label{thm:intro_chaos}
		Let us suppose that $D$ admits two central configurations which are not antipodal and non degenerate. Then, for sufficiently large inner energies, both the reflective and the refractive billiards are chaotic. \\
		More precisely, if the inner energy is sufficiently large it is possible to construct  first return maps,  restricted to a suitable invariant subset of initial conditions, which are topologically chaotic.  
	\end{theorem}
The main idea underlying the theorem is the construction of a  \textit{topological conjugation} with the \textit{Bernoulli shift} acting on bi-infinite sequences of two symbols. To be more precise, let us consider the set $\mathbb L=\{1,2\}^\Z$, and the so-called Bernoulli shift 
\begin{equation}\label{eq:Bernoulli}
	\sigma: \mathbb L\to \mathbb L, \quad \un s=(s_k)_{k\in \Z}\mapsto \un s'=(s'_k)_{k\in\Z}=(s_{k+1})_{k\in \Z}: 
\end{equation} 
such mapping is the model topological chaotic system (see \cite{Dev_book, hasselblattkatok}), and can be used to prove the chaoticity of more general discrete maps. In particular, given a map $\mathcal F: X\to X$, if we find a continuous bijection $P:X\to \mathbb L$ such that the diagram 
	\begin{equation}
	\begin{tikzcd}
		X \arrow{r}{{\mathcal F}} \arrow{d}{P} & X \arrow{d}{P} \\
		\mathbb{L} \arrow{r}{\sigma}	& \mathbb{L}
	\end{tikzcd}
\end{equation}
commutes, that is, $\mathcal F$ is \textit{topologically conjugated} to $\sigma$, then $\mathcal F$ is itself chaotic. The proof of Theorem \ref{thm:intro_chaos} will then be based on the search of an invariant set of the first return map, as well as a conjugation between the latter and $\sigma.$ In practice, we will construct a one-to-one relation between bi-infinite sequences on $\mathbb L$ and suitable billiard trajectories, requiring, through a \textit{shadowing} argument (see also \cite{Bol2017}), that they are sufficiently close to juxtaposition of homothetic arcs. \\
Domains that satisfy the hypotheses of Theorem \ref{thm:intro_chaos} are called \emph{admissible}: an example is given by triaxial ellipsoids having at least one semi-axis different from the other two. The chaoticity of an ellipsoid both in the reflective and the refractive case is particularly interesting when compared with the classical Birkhoff case, where it is possible to prove the complete integrability of the system in the ellipsoidal case (see \cite{tabachnikov2002ellipsoids}). \\
A particular case of domain is represented by the centred sphere (i.e. with the centre coinciding with the origin): here, any point on its surface is a degenerate central configuration.  We will see in Section \ref{sec:setting} that in the centred sphere the dynamics reduces to a two-dimensional integrable billiard with circular boundary when the first initial conditions are chosen (see also \cite{deblasiterraciniOnSome}). Moreover, applying Theorem \ref{thm:stability_intro} to any point, one obtains $\Delta_{1\backslash2}=0$ in both refractive and reflective case: this fact is again not surprising, and translates in the presence of infinite equilibrium trajectories, degenerate from the point of view of the linear stability. In the Kepler reflective case, the problem of finding three-dimensional boundaries that, at least at high energies, lead to chaos, is complemented by the work of Takeuchi and Zhao \cite{takeuchi2023kustaanheimo} where the problem of \textit{integrability} is taken into account. In general, this overall investigation looks in the direction of characterising integrable boundaries, in the spirit of the well-known \textit{Birkhoff conjecture} for classical billiards (see \cite{birkhoff2020periodic,poritsky1950billard}).

	\section{Analytical setting}\label{sec:setting}
	Let us start the description of the dynamical models by taking a regular domain $\underline 0\in D\subset \R^3$, namely, an open and bounded subset of $\R^3$, containing the origin, with smooth boundary. Let us suppose that $D$ is star-convex with respect to the origin, define $S\eqdef \partial D$, and suppose that it is a smooth surface in $\R^3$, whose regularity is at least $C^3$. \\
	The surface $S$ represents the boundary of our billiards, with which the trajectories, either in the reflective or in the refractive case, will interact according to the considered model.\\ 
	
	\paragraph{\textbf{Kepler billiards: reflective case}}
	
	Let us consider an inner potential as in \eqref{eq:inner_pot}: the trajectories of zero energy associated to the potential $V_I$ are  solutions of the ODE \eqref{eq:inner_prob_intro} for suitable boundary conditions, and correspond to hyperbolic Keplerian arcs having energy $\mathcal H$ and a focus in the origin. \\
	In Kepler billiards, whenever an inner arc intersects the domain's boundary $\partial D$ it is reflected back, generating a new trajectory starting from the bouncing point and with an initial velocity vector which, while preserving the original tangent component, reverts the sign of the normal one (see Figure \ref{fig:riflessione}). 
	\begin{figure}
		\centering
		\begin{overpic}[width=0.4\linewidth]{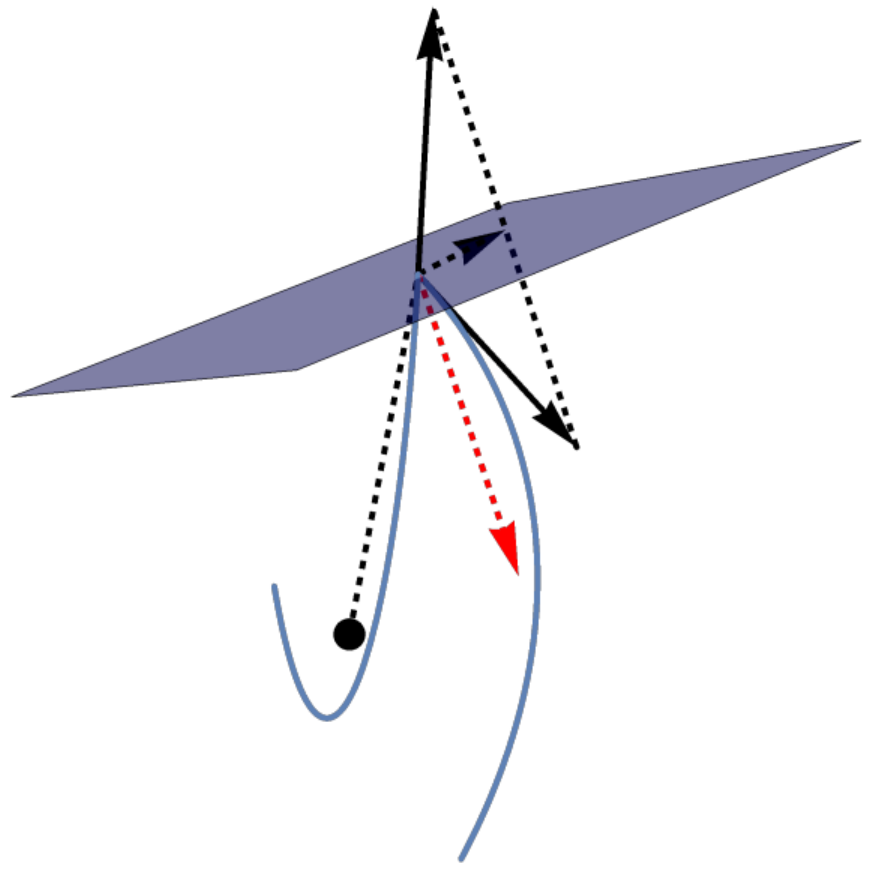}
			\put (20,25) {$z_1(t)$}
			\put (45, 100) {\rotatebox{0}{$w_1$}}
			\put (65,45) {$w_2$}
			\put (42,68) {$p$}
			\put (58,10) {$z_2(t)$}
			\put (90,75) {$T_p\partial D$}
		\end{overpic}
		\caption{Reflection rule in three dimensions; for the sake of clarity, only the tangent plane $T_p\partial D$ is drawn. An incoming arc $z_1$, with final velocity $w_1$, is reflected in a new arc $z_2$, whose initial velocity is given by $w_2$. The two velocities have the same tangent component with respect to $T_p\partial D$ and inverse normal component. The dashed black line represents the radial direction from the origin to $p$, while the red arrow is the inward-pointing normal vector to $\partial D$ in $p$. }
		\label{fig:riflessione}
	\end{figure}
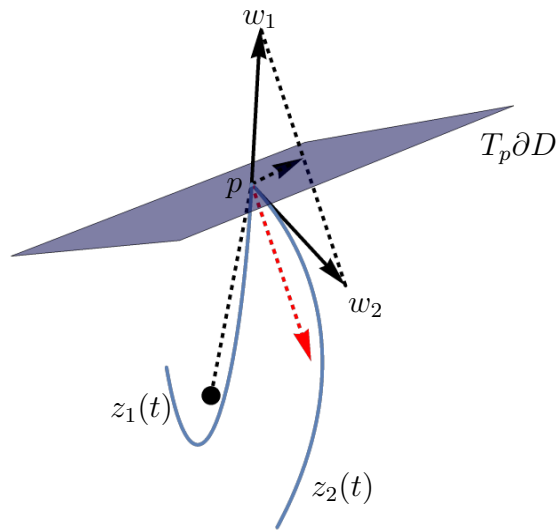
	The billiard trajectories result then in a concatenation of Keplerian arcs, related by a reflection rule, described by means of a discrete area-preserving mapping that keeps track of every time an inner arc intersects the boundary. \\
	To be more precise, let us suppose that a Keplerian arc starts from a point $p_0\in S$ with an inward-pointing initial velocity $w_0$, with $|w_0|=\sqrt{2(\mathcal H+\mu/|p_0|)}$. It represents the unique solution of the Cauchy problem
	\begin{equation}\label{eq:Cauchy_inner}
		\begin{cases}
			 z''(t)=\nabla V_I(z(t))\\
			z(0)=p_0, \  z'(0)=w_0
		\end{cases}
	\end{equation} 
	that intersects the surface $S$ again in a point $p_1\in S$ and with velocity vector $w_1'$ (see  Figure \ref{fig:primo ritorno Kep}), which can be decomposed into its tangent and normal part with respect to $T_{p_1}S$ as $w'_{1,t}+w'_{1,n}$. The next arc in the concatenation will then start from $p_1$ with velocity $w_1=w'_{1,t}-w'_{1,n}$. 	
	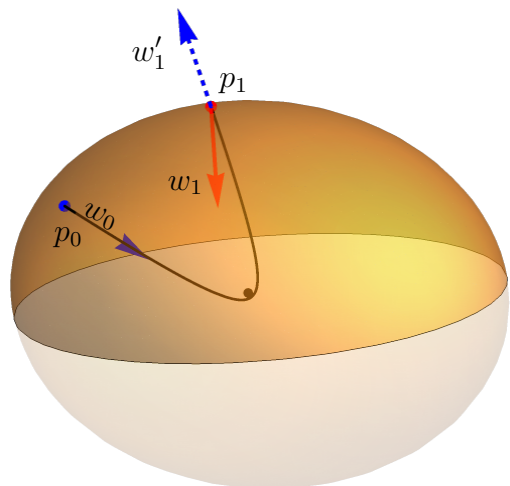
\begin{figure}
		\centering
		\begin{overpic}[width=0.4\linewidth]{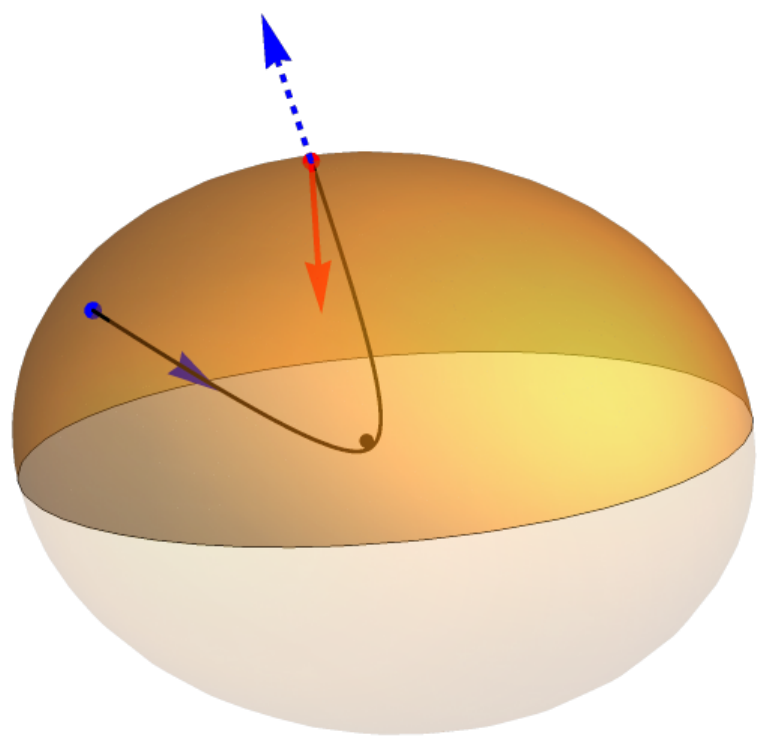}
			\put (10,50) {$p_0$}
			\put (15, 55) {\rotatebox{-20}{$w_0$}}
			\put (25,85) {\rotatebox{0}{$w_1'$}}
			\put (42,80) {$p_1$}
			\put (32,60) {$w_1$}
		\end{overpic}
		\caption{First return map for the reflective case: the trajectory starts in $p_0$ with velocity $w_0$ and meets again $\partial D$ in $p_1$ with velocity $w'_1$. The latter is reflected into $w_1$ and one has $F: (p_0, w_0)\mapsto (p_1, w_1)$. }
		\label{fig:primo ritorno Kep}
	\end{figure}
	It is clear that, unless $w'_1$ (and hence $w_1$) is tangent to $S$, the new velocity vector points again towards the interior of $D$, and a new inner Keplerian hyperbola can start with initial conditions $(p_1, w_1)$. \\
	The map $F:(p_0,w_0)\mapsto (p_1,w_1)$ is called \emph{first return map}, and, whenever well defined, its iterates give a complete portrait of the behaviour of the concatenation of reflected Keplerian arcs starting from $(p_0,w_0)$. 
	We highlight that the construction described above is a generalization in three dimensions of what happens in the two-dimensional case (see for example \cite{deblasiterracinirefraction}), where the tangent plane reduces to a simple direction in $\R^2$. As for the good definition of F, it depends on the geometric interaction of Keplerian arcs at a given energy with $S$: in Section \ref{ssec:stab_Kep}, we will delve deeper into this point. Let us notice a final fact about Kepler billiards, which motivates the necessity to pass carefully from the planar to the spatial case. Being $V_I$ a central force potential, the angular momentum is a conserved quantity along any arc, and, as a consequence, the Keplerian hyperbol\ae\ lie on planes defined by their initial conditions. It is then natural to ask whether the concatenations of a reflective Kepler billiard lie always on invariant planes, reducing the three-dimensional dynamics to a trivial generalization of the two-dimensional one. Actually, while the plane of the trajectories does not change along Keplerian hyperbol\ae, this could happen during the reflection. This is a consequence of two basic facts: 
	\begin{itemize}
		\item given initial conditions $(p_0,w_0)$, with $p_0\in S$ and $w_0\in\R^3$ pointing inside $D$, the solution of \eqref{eq:Cauchy_inner} lies on the plane generated by $w_0$ and $\overrightarrow{0p_0}$; 
		\item when  a velocity $w'$ is reflected into $w$ in a bouncing point $p\in S$, the new vector belongs to the plane generated by $w'$ and the normal unit vector to $S$ in $p$ (see again Figure \ref{fig:riflessione}). 
	\end{itemize}
	As a consequence, whenever in a bouncing point the normal and radial direction do not coincide, the new arc of hyperbola will belong to a different plane with respect to the reflected one and there are no invariant planes for the whole dynamics.  \\

	\paragraph{\textbf{Galactic refraction billiards: refractive case}}
	
	The dynamics induced on a refraction billiard  is more complex than the one depicted in the Kepler case, since here the particle is no longer confined inside $D$, but can exit from the domain, to which it is attracted again by a harmonic restoring central force. Given again $D\subset\R^3$ as an open and bounded domain containing the origin, let us now consider a new discontinuous potential
	\begin{equation}\label{eq:potentials}
		V: \R^3\setminus \{0\}\to \R, \quad V(z)=\begin{cases}
			V_I(z) &\text{if }z\in \overline D\\
			V_E(z)\quad &\text{if }z\notin\overline D  
		\end{cases}
	\end{equation} 
	where $V_E$ is defined in Eq. \eqref{eq:outer_pot}. In practice, we are considering  a Keplerian potential with positive energy inside $D$ coupled with a harmonic-type one  outside the domain; here, the constants $\mathcal E$ and $\mathcal H$  are such that $0<\mathcal E\leq\mathcal H$. In this case, the zero-energy trajectories inside $D$ are still Keplerian hyperbol\ae\  at energy $\mathcal H$, while the outer arcs are harmonic ellipses with energy $\mathcal E$ and centre at the origin. \\
	The complete trajectories in this case are concatenations of outer and inner arcs, connected by means of a refraction rule which, as in the case of the reflection law, keeps constant the tangent component velocity vectors during the transition through $S$. Let us suppose to have an outer arc intersecting $S$ again in the point $p$ with velocity $w_E\in\R^3$, refracted into an inner velocity vector $w_I$ (see Figure \ref{fig:rifrazione}).	 
	\begin{figure}
		\centering
		\begin{overpic}[width=0.5\linewidth]{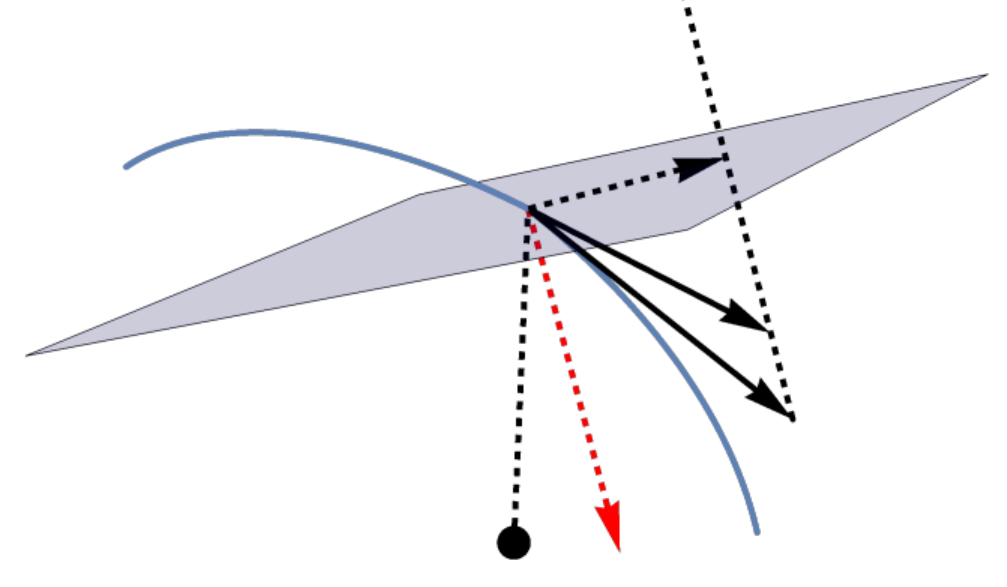}
			\put (10,44) {\rotatebox{20}{$z_E(t)$}}
			\put (48, 35) {\rotatebox{0}{$p$}}
			\put (70,0) {\rotatebox{0}{$z_I(t)$}}
			\put (78,25) {$w_E$}
			\put (80, 15) {$w_I$}
			\put (95, 45) {$T_p\partial D$}
		\end{overpic}
		\caption{Refraction rule; for the sake of clarity, only the tangent plane $T_p\partial D$ is drawn. The incoming outer arc $z_E$ intersects the boundary in $p$ with velocity $w_E$, which is refracted into $w_I$, initial velocity for the following inner arc $z_I$. The two velocity vectors have the same tangent component with respect to $T_p\partial D$.  The red arrow is the inward-pointing normal vector to $\partial D$ in $p$.}
		\label{fig:rifrazione}
	\end{figure}
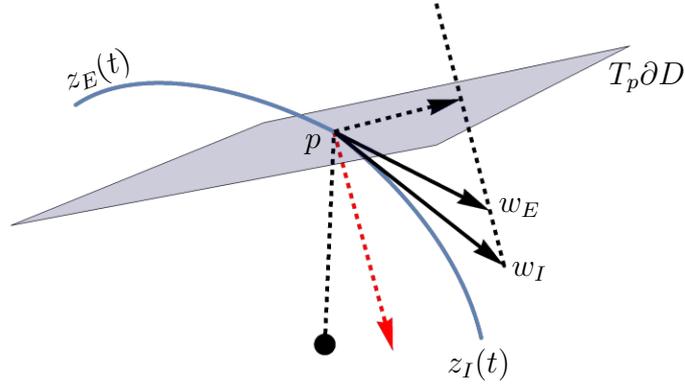
    By the energy conservation law, one has $|w_E|=\sqrt{2V_E(p)}$ and $|w_I|=\sqrt{2V_I(p)}$, where, by Eq.\eqref{eq:potentials}, $V_E(p)<V_I(p)$: a preservation of the tangent component from outside to inside $D$ translates then into an increasing of the normal one, that is, a deflection of the velocity vector. In particular, if $\alpha_E$ and $\alpha_I$ are respectively the angles of $w_E$ and $w_I$ with respect to the inward-pointing normal vector to $S$ in $p$, one has   
    \begin{equation}\label{eq:snell_angle}
    	\sqrt{V_I(p)}\sin\alpha_I=\sqrt{V_E(p)}\sin\alpha_E.  
    \end{equation}
	Eq. \eqref{eq:snell_angle} is a generalization for curved interfaces and non-constant potentials of the classical Snell's law for light rays, where the refraction indices depend on the two potentials at the transition point. An analogous way to describe Snell's refraction law, which will be useful in Section \ref{ssec:stab_hom_snell}, is the following: if $e_1,e_2$ is a basis of $T_pS$, one has that 
	\begin{equation}\label{eq:cons_tang}
		\scp{w_E, e_1}=\scp{w_I, e_1} \text{ and  }\scp{w_E, e_2}=\scp{w_I, e_2},. 
	\end{equation}		
	where the notation $\scp{\cdot,\cdot}$ denotes the usual scalar product in $\R^3$. 
	Clearly, whenever a velocity is refracted from inside to outside $D$, the normal component of the velocity decreases, and the angle will change again according to the relation \eqref{eq:snell_angle}. \\
	Note that, whenever the particle goes from an outer to the subsequent refracted arc, Eq.\eqref{eq:snell_angle} is solvable for any value of $\alpha_E$, while the contrary is possible only when $|\sin\alpha_I|\leq\sqrt{V_E(p)/V_I(p)}$. Geometrically, this translates into the presence of a \emph{critical angle} with respect to the normal direction, given by $\arcsin\left(V_E(p)/V_I(p)\right)$, under which an inner Keplerian arc must hit $S$ to be refracted. In general, we will say that trajectories of the refraction billiard are well defined when every inner arc can be refracted outside.\\
	As in the case of Kepler reflective billiard, it is possible to define a first return map in the refractive case as well. In this case, it will not consider every intersection of the trajectory with $S$, but only the ones occurring after a complete concatenation of an outer and a subsequent inner arc. \\
	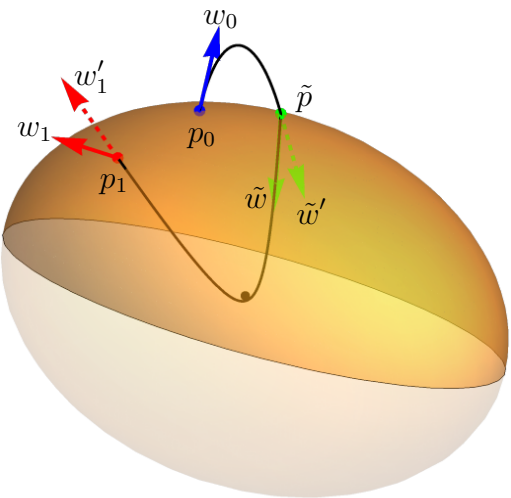
\begin{figure}
	\centering
	\begin{overpic}[width=0.4\linewidth]{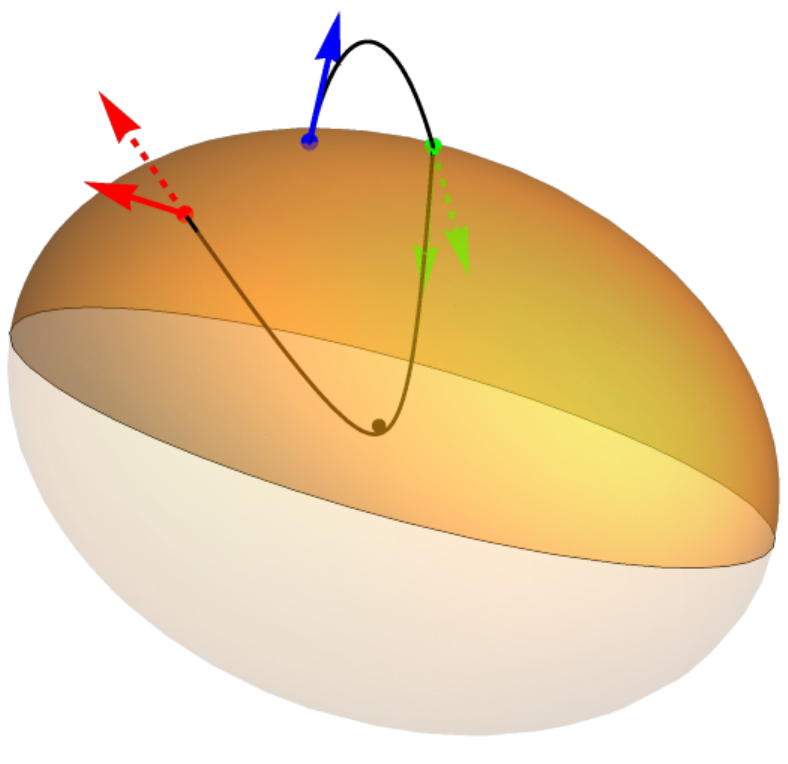}
		\put (37, 73) {\rotatebox{0}{$p_0$}}
		\put (40, 96) {\rotatebox{0}{$w_0$}}
		\put (58,80) {\rotatebox{0}{$\tilde p$}}
		\put (58,57) {$\tilde w'$}
		\put (48, 60) {$\tilde w$}
		\put (20, 64) {$p_1$}
		\put (15, 84) {$w_1'$}
		\put (4, 74) {$w_1$}
	\end{overpic}
	\caption{Refractive first return map: the outer arc starts with initial conditions $(p_0, w_0)$ and intersects $\partial D$ again in $(\tilde p, \tilde w')$. The velocity $\tilde w'$ is refracted into $\tilde w$, initial velocity for an inner arc which meets $\partial D$ again in $(p_1, w_1')$. The vector $w_1'$ is refracted into $w_1$ and one finally finds $G:(p_0, w_0)\mapsto (p_1, w_1)$.  }
	\label{fig:primo_rit_snell}
\end{figure}
	With reference to Figure \ref{fig:primo_rit_snell}, let us start with $(p_0,w_0)$ initial conditions for an outer arc, $p_0\in S$ and $w_0$ pointing outside $D$. The solution of the Cauchy problem associated to $V_E$ and with initial conditions as above is an elliptic arc which will intersect $S$ again, say in a point $\tilde p$ and with velocity $\tilde w'$, which will be refracted, according to Snell's law, into a new vector $\tilde w$. The pair $(\tilde p, \tilde w)$ will then be  the initial conditions of an inner Cauchy problem, whose solution is an unbounded hyperbola intersecting again $S$ in a point $p_1$ with velocity vector $w'_1$. If $w_1'$ is sufficiently transverse to $S$, it will be refracted into a new vector $w_1$, and a new concatenation of outer-inner arc can start from $(p_1, w_1)$. \\
	Whenever well defined, the map $G:(p_0,w_0)\mapsto(p_1,w_1)$ is the first return map associated to the refraction billiard, and, as in the reflective case, its iterates starting from different initial points encode our three-dimensional continuous dynamics into a sequence of discrete hitting points and velocities. Here, the problem of the good definition of the map is more delicate than in the previous case, since, as a critical angle is now present, the conditions on the inner arcs for the trajectory to continue are stricter. \\
	In this case, both the potentials $V_E$ and $V_I$ are centrally symmetric, therefore all the (outer and inner) arcs are planar. As in the reflective billiard, here a refraction always occurs on the plane generated by the original velocity and the normal vector to $S$ in the transition point, so that, whenever radial and normal direction do not coincide, the plane of the concatenation changes after refraction (see again Figure \ref{fig:rifrazione}). \\
	
	As customary in billiards theory, the dynamics in both models depends heavily on the geometry of $D$: it will be the aim of next sections to analyse such relation, starting from the existence of a particular kind of equilibrium trajectories. 
	
	\section{Homothetic equilibrium trajectories and their stability}\label{sec:stability}

	In both reflective and refractive billiards, there exists a particular class of equilibrium trajectories whose existence depends only on the local properties of $S$. These are called \emph{homothetic trajectories}, and are composed by straight segments along particular directions of $\R^3$, defined by the so-called \textit{central configurations}. This section will address the problem of their existence and stability. We will make use of some basic facts regarding outer harmonic and inner Keplerian arcs; for the sake of brevity, some of the technical lemmas used here are proved in Appendix \ref{sec:app}. 
	\begin{definition}\label{def:cc}
	A point $\bar p \in S$ is called \emph{central configuration}\footnote{In the present paper we assume a star-convexity hypothesis on $D$, not present in \cite{deblasiterracinirefraction, deblasiterraciniOnSome, deblasiterracinibarutelloChaotic}. In more general domains, to define a central configuration we also require that the half-line in the direction of $\bp$ intersects $S$ only once.} (in short, c.c.) for the domain $D$ if it is orthogonal to $T_{\bar p}S$. . 
	\end{definition}
	In local coordinates, central configurations correspond to critical points of the distance function restricted to points in $S$: let us consider $\Gamma: U\to S$, $U\subset\R^2$, $(u,v)\mapsto \Gamma(u,v)\in S$ a local parametrisation of $S$ around $\bar p$, such that $\Gamma(\bu, \bv)=\bar p$. Then 
	\begin{equation}\label{eq:perp}
		\bar p \perp T_{\bar p}S\ \Longleftrightarrow \nabla |\Gamma(u,v)|_{(\bu,\bv)}=\underline 0. 
	\end{equation}
	Eq. \eqref{eq:perp} allows us to define a further property for a central configuration. 
	\begin{definition}
		Let $\bar p$ a central configuration, and $\Gamma:U\to S$ a local parametrisation of $S$ around $\bar p$, $\Gamma(\buv)=\bp$. The c.c. is termed \emph{non-degenerate} if $(\buv)$ is a non-degenerate critical point of $|\Gamma(u,v)|$, namely, its Hessian computed in $(\buv)$ admits two nonzero eigenvalues. 
	\end{definition}

\begin{remark}\label{rem:par}
	Let $\Gamma: U\to S$ be a local chart of $S$ around $\bp$. For the sake of convenience and without loss of generality, we will assume that $\Gamma$ satisfies the following conditions, being $\Gamma_u(u, v)$ and $\Gamma_v(u, v)$ the coordinate tangent vectors to $S$: 
	\begin{equation*}
	\begin{aligned}
		&\bullet\  &&\Gamma(\bu, \bv)=\bp\text{ for some }(\bu, \bv)\in U; \\
		&\bullet\ &&\Gamma_u(\bu, \bv)\text{ and }\Gamma_v(\bu, \bv)\text{ are orthogonal unit vectors, and, if }\bp\text{ is not }\\&\  &&\text{umbilical, they are parallel to the principal directions of }S\text{ in }\bp; \\
		&\bullet\ &&\overline N=\Gamma_u(\bu, \bv)\times\Gamma_v(\bu, \bv)\text{ points inside }D. 
	\end{aligned}
\end{equation*}
 We stress that it is always possible to find such a parametrization, possibly rescaling $u$ and $v$ and with a change of basis on $T_{\bp}S$. \\ 
 Being $\bp$ a central configuration, it results $\bp =- |\bp| \overline N$: calling $\ku$ and $\kd$ the principal curvatures of $S$ respectively in the direction of $\Gamma_u(\bu, \bv)$ and $\Gamma_v(\bu, \bv)$,  by direct computation one finds
\begin{equation}\label{eq:par_giusta}
	\scp{\bar p,\Gamma_{uu}(\bar u,\bar v)}=-|\bar p|\kappa_1, \ \scp{\bar p, \Gamma_{uv}(\bar u,\bar v)}=\scp{\bar p, \Gamma_{vu}(\bar u,\bar v)}=0, \ \scp{\bar p, \Gamma_{vv}(\bar u,\bar v)}=- |\bar p|\kappa_2. 
\end{equation}
As a consequence, 
\begin{equation}	\label{eq:der_par}
	\begin{aligned}
		&\partial_u|\Gamma(\buv)|=\partial_v|\Gamma(\buv)|= \partial_{uv}|\Gamma(\buv)|=0, \\
		& \partial_u^2|\Gamma(\buv)|=\frac{1}{|\bar p|}\left(1-|\bar p|\kappa_1\right),\  \partial_v^2|\Gamma(\buv)|=\frac{1}{|\bar p|}\left(1-|\bar p|\kappa_2\right) 
	\end{aligned}
\end{equation}
and then, whenever $\bp$ is non degenerate, $\partial_u^2|\Gamma(\buv)|$ and $\partial_v^2|\Gamma(\buv)|$ are different from zero.
\end{remark}

	Since $V_I$ and $V_E$ are centrally symmetric, both the inner and outer dynamics admit radial solutions (possibly collisional in the inner case, see Theorems \ref{thm:outer problem}, \ref{thm:existence_inner});  whenever hitting the boundary in a c.c., such solutions are either reflected back into themselves (reflective case, see Figure \ref{fig:homothetics}, left) or they are not deflected by Snell's law (refractive case, see Figure \ref{fig:homothetics}, right). As a consequence, one can state the following Proposition. 
	
	\begin{proposition}
		If $\bp\in S$ is a central configuration, both the reflective and refractive billiard admit  equilibrium trajectories along the direction of $\bp$, called \emph{homothetic}. In the reflective case, they are collision-ejection arcs infinitely many times reflected into themselves; in the refractive case, they are concatenations of outer homothetic arcs and inner collisional ones. 
	\end{proposition} 
 The definition of collision-ejection inner arcs is given in Appendix \ref{sec:app}, and refers to particular generalised solutions associated to the Kepler problem which, after hitting the origin, are reflected; their existence is proved by a \textit{regularization} technique, see Theorem \ref{thm:inner problem 2}. 
In terms of first return maps, homothetic equilibrium trajectories translate into fixed points, corresponding to the pair $(\bp,-\sqrt{V_I(\bp)}\bp/|\bp|)$ for the reflective case and $(\bp, \sqrt{2 V_E(\bp)}\bp/|\bp|)$ for the refractive one; by analogy with the corresponding trajectories, such points are called \emph{homothetic}. \\
As a consequence of Theorems \ref{thm:outer problem} and \ref{thm:inner problem 1}, along with the differentiable dependence of Cauchy problems on initial conditions,  both the reflective and refractive first return maps are well defined and differentiable in a neighbourhood of a homothetic fixed point: it is then natural to investigate their linear stability. To achieve this, we shall make use of the notion of \emph{Jacobi length}. Let us then take $V\subset S$ as a neighbourhood of $\bp$ such that $V\subset \Gamma(U)$; if it is sufficiently small and the inner energy $\mathcal H$ large enough, for any $p_0, p_1\in U$ there are exactly one inner and one outer arcs connecting them respectively at energies $\mathcal H$ and $\mathcal E$, where, in the inner case, the uniqueness is intended in a suitable homotopy class. Let denote these arcs respectively with $z_I(\cdot; p_0, p_1)$ and $z_E(\cdot; p_0, p_1)$. We can then define the \textit{inner} and \textit{outer distances} between the two points as the Jacobi lengths of the (unique) arcs connecting them, namely, 
\begin{equation}
	\begin{aligned}
	d_E(p_0,p_1)=\mathcal L_E(z_E(\cdot;p_0, p_1))=\int_0^T| z'_E(t; p_0,p_1)|\sqrt{V_E(z(t; p_0,p_1))}dt,\\ 
	d_I(p_0,p_1)=\mathcal L_I(z_I(\cdot;p_0, p_1))=\int_0^T| z'_I(t; p_0,p_1)|\sqrt{V_I(z(t; p_0,p_1))}dt,
	\end{aligned}
\end{equation}
where, with an abuse of notation, the time $T>0$ is such that either $z_E(T; p_0,p_1)=p_1$ or $z_I(T; p_0,p_1)=p_1$. By the uniqueness of the arcs, the above quantities are always well defined and depend only on the endpoints; in the case of a collision-ejection solution one shall make use of a regularization method. The main properties of $d_E$ and $d_I$ are described in Appendix \ref{sec:app}. \\
Let us now express the distances in terms of local coordinates, taking $(u_0,v_0), (u_1,v_1)\in U$ such that $p_i=\Gamma(u_i,v_i)$, $i=1,2$, and  defining the functions 
\begin{equation}\label{eq:gen}
	\begin{aligned}
		\mathcal S_{E\backslash I}: U\times U\to \R, \quad (u_0, v_0, u_1, v_1)\mapsto d_{E\backslash I}\left(\Gamma(u_0, v_0), \Gamma(u_1, v_1)\right). 
	\end{aligned}
\end{equation} 
The above  functions inherit the regularity of the surface $S$, and, dropping for brevity the dependence of $z_{E\backslash I}$ on the points, it is possible to prove (see Appendix \ref{sec:app}) that 
\begin{equation}\label{eq:der1S}
\begin{aligned}
	&\partial_{1,u}\mathcal S_{E\backslash I}(u_0,v_0, u_1, v_1)=-\sqrt{V_{E\backslash I}(\Gamma(u_0,v_0))}\scp{\frac{z'_{E\backslash I}(0)}{|z'_{E\backslash I}(0)|}, \Gamma_u(u_0,v_0)}\\
	&\partial_{1,v}\S_{E\backslash I}(u_0,v_0, u_1, v_1)=-\sqrt{V_{E\backslash I}(\Gamma(u_0,v_0))}\scp{\frac{z'_{E\backslash I}(0)}{|z'_{E\backslash I}(0)|}, \Gamma_v(u_0,v_0)}\\
	&\partial_{2,u}\S_{E\backslash I}(u_0,v_0, u_1, v_1)=\sqrt{V_{E\backslash I}(\Gamma(u_1,v_1))}\scp{\frac{z'_{E\backslash I}(T)}{|z'_{E\backslash I}(T)|}, \Gamma_u(u_1,v_1)}\\
	&\partial_{2,v}\S_{E\backslash I}(u_0,v_0, u_1, v_1)=\sqrt{V_{E\backslash I}(\Gamma(u_1,v_1))}\scp{\frac{z'_{E\backslash I}(T)}{|z'_{E\backslash I}(T)|}, \Gamma_v(u_1,v_1)},
\end{aligned}	
\end{equation} 
where the subscripts $1$ or $2$ refer respectively to the first or second pair $(u,v)$.\\
In analogy to the classical theory of billiards, and dynamical systems in general, we will call the quantities $\mathcal S_{E\backslash I}$ the \textit{generating functions} of our dynamics: they will have a crucial role in the forthcoming analysis. In particular, as a direct consequence of \eqref{eq:der1S}, both the reflection and the refraction laws can be expressed in terms of derivatives of $\S_{E\backslash I}$: 
\begin{itemize}
	\item for the reflective case, let us take $p_0, p_1, p_2\in\Gamma(U)$, $p_i=\Gamma(u_i, v_i)$, $i=0,1,2$, connected by a concatenation of arcs from $p_0$ to $p_2$ passing through $p_1$: the reflection holds in $p_1$ if and only if
	\begin{equation}\label{eq:riflS}
		\begin{cases}
		\partial_{2,u}\S_I(u_0,v_0, u_1, v_1)+\partial_{1,u}\S_I(u_1,v_1, u_2, v_2)=0, \\ \partial_{2,v}\S_I(u_0,v_0, u_1, v_1)+\partial_{1,v}\S_I(u_1,v_1, u_2, v_2)=0. 
		\end{cases}
	\end{equation}
	\item In  the refractive case, let us suppose that $p_0$ and $p_1$ are connected by an outer arc, while $p_1$ and $p_2$ by an inner one. Then, Snell's law in $p_1$ is equivalent to 
	\begin{equation}\label{eq:refrS}
		\begin{cases}
			\partial_{2,u}\S_E(u_0,v_0, u_1, v_1)+\partial_{1,u}\S_I(u_1,v_1, u_2, v_2)=0, \\ \partial_{2,v}\S_E(u_0,v_0, u_1, v_1)+\partial_{1,v}\S_I(u_1,v_1, u_2, v_2)=0. 
		\end{cases}
	\end{equation} 
\end{itemize}
In the next Sections, we will see how to use the above relations to determine the linear stability of homothetic equilibrium trajectories in both reflective and refractive billiards. 

\subsection{Stability of the homothetics in the reflective case} \label{ssec:stab_Kep}
 Let us start by expressing the first return map $F$ in a suitable set of one-dimensional variables: with reference to Figure \ref{fig:primo ritorno Kep}, once we fix $\H$ the initial and final conditions $(p_0, w_0)$, $(p_1, w_1)$ are uniquely determined by the quadruplets $(u_0, v_0, x_0, y_0),(u_1, v_1, x_1, y_1)\in U\times \R^2 $ such that $p_i=\Gamma(u_i, v_i)$, $x_i=\scp{w_i/|w_i|, \Gamma_u(u_i, v_i)}$ and $y_i=\scp{w_i/|w_i|, \Gamma_v(u_i, v_i)}$. In this set of coordinates, the homothetic fixed point  $(\bp,-\sqrt{V_I(\bp)}\bp/|\bp|)$   corresponds to $(\bu, \bv, 0,0)$. We can then construct the function 
\begin{equation}
	\tilde F : (u_0, x_0, v_0, y_0)\mapsto (u_1, x_1, v_1, y_1), 
\end{equation}
expressing the first return map locally around $(\bu,0, \bv, 0)$. We observe that, in view of 
Eqs. \eqref{eq:der1S}, $\tilde F$ is implicitly defined by the equation $\Phi(u_0, x_0, u_1, x_1, v_0, y_0, v_1, y_1)=\underline 0$, where $\Phi$ is a multi-valued function whose components are given by
\begin{equation}\label{eq;Phi}
	\begin{aligned}
		\partial_{1,u}\S_{I}(u_0,v_0, u_1, v_1)+\sqrt{V_{I}(\Gamma(u_0,v_0))}x_0\\
		\partial_{2,u}\S_{I}(u_0,v_0, u_1, v_1)-\sqrt{V_{I}(\Gamma(u_1,v_1))}x_1\\
		\partial_{1,v}\S_{I}(u_0,v_0, u_1, v_1)+\sqrt{V_{I}(\Gamma(u_0,v_0))}y_0\\		
		\partial_{2,v}\S_{I}(u_0,v_0, u_1, v_1)-\sqrt{V_{I}(\Gamma(u_1,v_1))}y_1. 
		\end{aligned}
\end{equation}
Imposing $\bar q=(\bu, 0, \bv, 0)$, it is easy to verify that $\Phi(\bar q,\bar q)=0$: our aim now is to use the implicit function theorem to find $D\tilde F(\bar q)$, and then study its spectral properties.\\
The first step consists in finding the second derivatives of $\S_I$ in $\bar q$: the procedure we use is a generalization in three dimensions what has been done in \cite{deblasiterracinirefraction}, where the interested reader can find more details. \\
Recalling the choice of the local chart $\Gamma$ described in Remark \ref{rem:par} and making use of a suitable regularization reasoning, one finds the following Lemma, whose proof is given in Appendix \ref{sec:app}.
\begin{lemma}\label{lem:der_SI_hom}
	The second derivatives of $\S_I$ in the homothetic point $(\bu,\bv, \bu, \bv)$ are given by \begin{equation}\label{eq:der2S}
		\begin{aligned}
			&\partial_{1,u}^2\S_I(\bu,\bv, \bu, \bv)=\partial_{2,u}^2\S_I(\bu,\bv, \bu, \bv)=\mathcal I_0+\eta_1\\
			&\partial_{1, v}^2\S_I(\bu,\bv, \bu, \bv)=\partial_{2,v}^2\S_I(\bu,\bv, \bu, \bv)=\mathcal I_0+\eta_2\\
			&\partial_{\substack{1,u\\2,u}}\S_I(\bu,\bv, \bu, \bv)=\partial_{\substack{2,u\\ 1,u}}\S_I(\bu,\bv, \bu, \bv)=-\mathcal I_0\\
			&\partial_{\substack{1,v\\2,v}}\S_I(\bu, \bv, \bu, \bv)=\partial_{\substack{2,v\\1,v}}\S_I(\bu, \bv, \bu, \bv)=-\mathcal I_0\\
			&\partial_{\substack{1,u\\1,v}}\S_I(\bu, \bv, \bu, \bv)=\partial_{\substack{1,u\\2,v}}\S_I(\bu, \bv, \bu, \bv)=\partial_{\substack{2,u\\1,v}}\S_I(\bu, \bv, \bu, \bv)=\partial_{\substack{2,u\\2,v}}\S_I(\bu, \bv, \bu, \bv)=0\\
			&\mathcal I_0=-\frac{\mu}{4 |\bar p|^2\sqrt{V_I(\bar p)}}, \ \eta_1=\sqrt{V_I(\bar p)}\left(\frac{1}{|\bar p|}-\kappa_1\right), \ \eta_2=\sqrt{V_I(\bar p)}\left(\frac{1}{|\bar p|}-\kappa_2\right). 
		\end{aligned}
	\end{equation}
\end{lemma}
Let us then return to Eqs. \eqref{eq;Phi}, and define
\begin{equation}
	D_0=D_{(u_0,x_0, v_0, y_0)}\Phi(\bar q ,\bar q),\quad D_1=D_{(u_1, x_1, v_1, y_1)}\Phi(\bar q ,\bar q), 
\end{equation} 
	respectively the Jacobian matrices of $\Phi$ with respect to the first and second quadruplet, computed in $(\bar q,\bar q)$. We will prove that $D_1$ is nonsingular, confirming that $\tilde F: (u_0, x_0,v_0, y_0)\mapsto (u_1, x_1, v_1, y_1)$ is well defined locally around $\bar q$, and therefore obtaining  an explicit formula for its Jacobian matrix, given by 
	\begin{equation}
		D\tilde F(\bar q)= -\left(D_1\right)^{-1} \ D_0. 
	\end{equation} 
	By the choice of the parameterisation $\Gamma$, one has that all the above matrices are block diagonal, and by direct computation, one finds
	\begin{equation}
		\det(D_1)=V_I\left(\Gamma(\bu,\bv)\right)\partial_{\substack{1,u\\2,u}}\S_I(\bu,\bv,\bu,\bv)\partial_{\substack{1,v\\2,v}}\S_I(\bu,\bv,\bu,\bv)=
		 V_I\left(\Gamma(\bu,\bv)\right)\mathcal I_0^2\neq0. 
	\end{equation}
	One can then proceed and compute
	\begin{equation}
		\begin{aligned}
			D\tilde F(\bar q) & = \begin{pmatrix}
				C_u & \boldsymbol{0}\\
				\boldsymbol{0} & C_v
			\end{pmatrix}\\
			C_u & =\begin{pmatrix}
				-\frac{\partial_{u_0}^2S_I(\bu,\bv,\bu,\bv)}{\partial_{u_0u_1}S_I(\bu,\bv,\bu,\bv)} & -\frac{\sqrt{V_I\left(\Gamma(\bu,\bv)\right)}}{\partial_{u_0u_1}S_I(\bu,\bv,\bu,\bv)}\\[20pt]
				-\frac{\partial_{u_0}^2S_I(\bu,\bv,\bu,\bv)\partial_{u_1}^2S_I(\bu,\bv,\bu,\bv)}{\partial_{u_0u_1}S_I(\bu,\bv,\bu,\bv)\sqrt{V_I\left(\Gamma(\bu,\bv)\right)}}+\frac{\partial_{u_0u_1}S_I(\bu,\bv,\bu,\bv)}{\sqrt{V_I\left(\Gamma(\bu,\bv)\right)}} & -\frac{\partial_{u_1}^2S_I(\bu,\bv,\bu,\bv)}{\partial_{u_0u_1}S_I(\bu,\bv,\bu,\bv)}\end{pmatrix}\\[20pt]
			C_v & =\begin{pmatrix}
				-\frac{\partial_{v_0}^2S_I(\bu,\bv,\bu,\bv)}{\partial_{v_0v_1}S_I(\bu,\bv,\bu,\bv)} & -\frac{\sqrt{V_I\left(\Gamma(\bu,\bv)\right)}}{\partial_{v_0v_1}S_I(\bu,\bv,\bu,\bv)}\\[20pt]
				-\frac{\partial_{v_0}^2S_I(\bu,\bv,\bu,\bv)\partial_{v_1}^2S_I(\bu,\bv,\bu,\bv)}{\partial_{v_0v_1}S_I(\bu,\bv,\bu,\bv)\sqrt{V_I\left(\Gamma(\bu,\bv)\right)}}+\frac{\partial_{v_0v_1}S_I(\bu,\bv,\bu,\bv)}{\sqrt{V_I\left(\Gamma(\bu,\bv)\right)}} & -\frac{\partial_{v_1}^2S_I(\bu,\bv,\bu,\bv)}{\partial_{v_0v_1}S_I(\bu,\bv,\bu,\bv)}
			\end{pmatrix}. 
		\end{aligned}
\end{equation}\\

It is possible to check that, as expected while dealing with area-preserving discrete maps, $\det(C_u)=\det{(C_v)}=1$, and  the characteristic polynomial associated to $D\tilde F(\bar q)$ can be factorised as $p(\lambda)=p_u(\lambda)p_v(\lambda)$, where $p_{u\backslash v}(\lambda)=\lambda^2-tr(C_{u\backslash v})+1$. The matrix $D\tilde F(\bar q)$ admits then two pairs of eigenvalues $(\lambda_u, \lambda_u^{-1})$, $(\lambda_v, \lambda_v^{-1})$ corresponding respectively to variations of $(u_0, x_0)$ from $(\bu, 0)$ and $(v_0, y_0)$ from $(\bv, 0)$. Such eigenvalues could be either real or complex, depending on the sign of the discriminants
\begin{equation}
	\Delta_u=tr(C_u)^2-4, \quad \Delta_v=tr(C_v)^2-4: 
\end{equation} 
 taking into account Eqs. \eqref{eq:der2S}, one finds 
 \begin{equation}
 	\begin{aligned}
 		&\Delta_u=\frac{4 \eta_1(2 \mathcal I_0+\eta_1)}{\mathcal I_0^2}=\frac{4}{\mathcal I_0^2}\sqrt{V_I(\bar p)}\left(\frac{1}{|\bar p|}-\kappa_1\right)\left[\sqrt{V_I(\bar p)}\left(\frac{1}{|\bar p|}-\kappa_1\right)-\frac{\mu}{2|\bar p|^2\sqrt{V_I(\bar p)}}\right],\\
 		&\Delta_v=\frac{4}{\mathcal I_0^2}\sqrt{V_I(\bar p)}\left(\frac{1}{|\bar p|}-\kappa_2\right)\left[\sqrt{V_I(\bar p)}\left(\frac{1}{|\bar p|}-\kappa_2\right)-\frac{\mu}{2|\bar p|^2\sqrt{V_I(\bar p)}}\right]. 
 	\end{aligned}
 \end{equation} 
The linear stability of $(\bu, 0, \bv, 0)$ (and, as a consequence, of the homothetic trajectory) then depends both on the geometric properties of $S$ in $\bp$, in terms of principal curvatures, and on the physical parameters involved in the problem, namely, the inner energy $\H$ and the mass parameter $\mu$. In particular, it holds that:
\begin{itemize}
	\item whenever $\Delta_u>0, \Delta_v>0$, every eigenvalue is real, and  the homothetic is a saddle in every direction; 
	\item if $\Delta_u\cdot \Delta_v<0$, then we have a saddle-centre configuration; 
	\item if $\Delta_u<0$ and $\Delta_v<0$, then we can conclude that the homothetic trajectory is a centre, thus linearly stable. 
\end{itemize}
\begin{remark}
Let us notice that, if $\Delta_u=0$ or $\Delta_v=0$, we are in a degenerate case, where we can not deduce the linear stability of our equilibrium trajectory. This is for example what happens when $\kappa_i=|\bp|^{-1}$ for either $i=1$ or $i=2$: geometrically, in this case the normal section of $S$ in $\bp$ either in the direction of $\Gamma_u(\bu, \bv)$ or in the direction of $\Gamma_v(\bu, \bv)$ is locally a circle of radius $|\bp|$. \\
If in particular $S$ is a sphere of radius $R$, it results that in every point $\ku=\kd=1/R$: in this case, every radial direction corresponds to a degenerate homothetic fixed point for the first return map, forming a two-dimensional invariant subspace of the four-dimensional phase space. 
\end{remark}

The previous conclusions prove Theorem \ref{thm:stability_intro}, in the reflective case, which is the first fundamental result of this paper.

\subsection{Stability of the homothetics in the refractive case}\label{ssec:stab_hom_snell}

In the refractive case we use the same procedure described in Section \ref{ssec:stab_Kep}, with two main differences: the first is that, here, we have to take into account composite concatenations of outer and inner arcs, then the dimensionality of the problem increases. Secondly, we will have to consider refraction Snell's law, whose variational characterization in terms of generating functions is given in Eq. \eqref{eq:refrS}. \\
Let us recall the first return map $G$ associated to the refractive billiard (see again Figure \ref{fig:primo_rit_snell}): given $(p_0, w_0)$ as the initial conditions of an outward-pointing harmonic arc, it returns $(p_1, w_1)$ final conditions obtained after a concatenation outer-inner arc, including refractions, passing through the point $\tilde p\in S$. As in the reflective case, whenever $p_0, \tilde p, p_1\in V$, we can describe the local dynamics by means of a set of four one-dimensional coordinates, describing points on $S$ and related velocities. In particular, let us consider $(u_0, v_0), (\tilde u, \tilde v), (u_1, v_1)\in U$ and $(x_0, y_0), (x_1, y_1)\in \R^2$ such that 
\begin{equation}
	\begin{aligned}
	&p_0=\Gamma(u_0, v_0), \ \tilde p=\Gamma(\tilde u, \tilde v), \ p_1=\Gamma(u_1, v_1), \\
	&x_i=\scp{\frac{w_i}{|w_i|}, \Gamma_u(u_i, v_i)}, \ y_i=\scp{\frac{w_i}{|w_i|}, \Gamma_v(u_i, v_i)}, \ i=0,1. 
	\end{aligned}
\end{equation}
By virtue of Eqs. \eqref{eq:der1S} and \eqref{eq:refrS}, whenever such coordinates represent a concatenation of refracted arcs, they have to satisfy the relations 
\begin{equation}\label{eq:Psi_comp}
	\begin{aligned}
		&\partial_{1,u}\S_E(u_0,v_0, \tilde u, \tilde v)+\sqrt{V_E(\Gamma(u_0,v_0))}x_0=0\\
		&\partial_{2,u}\S_E(u_0,v_0, \tilde u, \tilde v)+\partial_{1, u}\S_I(\tilde u,\tilde v, u_1, v_1)=0\\
		&\partial_{2,u}\S_I(\tilde u,\tilde v, u_1, v_1)-\sqrt{V_E(\Gamma(u_1,v_1))}x_1=0\\
		&\partial_{1,v}\S_E(u_0,v_0, \tilde u, \tilde v)+\sqrt{V_E(\Gamma(u_0,v_0))}y_0=0\\
		&\partial_{2,v}\S_E(u_0,v_0, \tilde u, \tilde v)+\partial_{1,v}\S_I(\tilde u,\tilde v, u_1, v_1)=0\\
		&\partial_{2,v}\S_I(\tilde u,\tilde v, u_1, v_1)-\sqrt{V_E(\Gamma(u_1,v_1))}y_1=0, 
	\end{aligned}
\end{equation}
where the third and last equations are a direct consequences of Eq. \eqref{eq:cons_tang}.  Eqs. \eqref{eq:Psi_comp} define implicitly the refractive first return map in the new coordinates, and can be summarized in the form $\Psi(u_0, x_0, v_0, y_0, \tilde u, \tilde v, u_1, x_1, v_1, y_1)=\underline 0$, where the components of $\Psi$ are the l.h.s. of \eqref{eq:Psi_comp}. The homothetic equilibrium solution in the refractive case corresponds to the coordinates $u_0=\tilde u=u_1=\bu$, $v_0=\tilde v=v_1=\bv$, $x_0=x_1=y_0=y_1=0$: again, we will use the implicit function theorem to compute the Jacobian matrix of the first return map 
\begin{equation}\label{eq:frmG}
	\tilde G: (u_0, x_0, v_0, y_0)\mapsto (u_1, x_1, v_1, y_1)
\end{equation}  
in the fixed point $\bar q=(\bu, 0, \bv, 0)$. Since derivatives of the components of $\Phi$ are involved, we have to use again the second derivatives of the generating functions $\S_{E\backslash I}$ computed in the homothetic trajectories: while the ones of $\S_I$ have been already given in Lemma \ref{lem:der_SI_hom}, those of $\S_E$ are listed below. 
\begin{lemma}\label{lem:der_SE_hom}
	The second derivatives of $\S_E$ in the homothetic point $(\bu, \bv, \bu, \bv)$ are given by 
	\begin{equation}\label{eq:der2SE}
		\begin{aligned}
			&\partial_{1,u}^2\S_E(\bu, \bv, \bu, \bv)&&=\partial_{2,u}^2 \S_E(\bu,\bv,\bu, \bv)=\mathcal E_0+\epsilon_1\\
			&\partial_{1,v}^2\S_E(\bu, \bv, \bu, \bv)&&=\partial_{2,v}^2 \S_E(\bu,\bv,\bu, \bv)=\mathcal E_0+\epsilon_2\\
			&\partial_{\substack{1,u\\2,u}}\S_E(\bu, \bv, \bu, \bv)&&=\partial_{\substack{2,u\\1,u}} \S_E(\bu,\bv,\bu, \bv)=-\mathcal E_0\\
			&\partial_{\substack{1,v\\2,v}}\S_E(\bu, \bv, \bu, \bv)&&=\partial_{\substack{2,v\\1,v}} \S_E(\bu,\bv,\bu, \bv)=-\mathcal E_0\\
			&\partial_{\substack{1,u\\1,v}}\S_E(\bu, \bv, \bu, \bv)&&=\partial_{\substack{1,v\\1,u}}\S_E(\bu, \bv, \bu, \bv)=\partial_{\substack{1,u\\2,v}}\S_E(\bu, \bv, \bu, \bv)=\partial_{\substack{2,v\\1, u}}\S_E(\bu, \bv, \bu, \bv)\\
			&\quad &&=\partial_{\substack{1,v\\2, u}}\S_E(\bu, \bv, \bu, \bv) =\partial_{\substack{2,u\\1,v}}\S_E(\bu, \bv, \bu, \bv)=\partial_{\substack{2,u\\2,v}}\S_E(\bu, \bv, \bu, \bv)\\&\quad&&=\partial_{\substack{2, v\\2,u}}\S_E(\bu, \bv, \bu, \bv)=0, 
		\end{aligned}
	\end{equation}
where 
\begin{equation}\label{eq:E0eps}
	\mathcal E_0=\frac{\mathcal E}{2|\bar p|\sqrt{V_E(\bar p)}},\ \epsilon_1=\sqrt{V_E(\bar p)}\left(\kappa_1-\frac{1}{|\bar p|}\right), \ \epsilon_2=\sqrt{V_E(\bar p)}\left(
	\kappa_2-\frac{1}{|\bar p|}\right). 
\end{equation}
\end{lemma}
The proof of the above result is analogous to the one already presented in \cite{deblasiterracinirefraction}, and is resumed in  Appendix \ref{sec:app}.\\
As in the reflective case, we will start by computing 
\begin{equation}
	A_0=D_{(u_0, x_0, v_0, y_0)}\Psi(\bar q, \bu, \bv, \bar q),\  A_1=D_{(\tilde u, \tilde v, u_1, x_1, v_1, y_1)}\Psi(\bar q, \bu, \bv, \bar q): 
\end{equation}
if $\det(A_1)\neq 0$, then, locally around the homothetic point,  the function 
\begin{equation}
	\tilde G_1: (u_0, x_0, v_0, y_0)\mapsto(\tilde u, \tilde v, u_1, x_1, v_1, y_1), 
\end{equation}
whose last four component represent precisely the first return map $\tilde G$ as in Eq. \eqref{eq:frmG}, is well defined. \\
Again, by direct computations and using Eqs. \eqref{eq:Psi_comp}, \eqref{eq:E0eps}, \eqref{eq:der2S}, one finds 
\begin{equation}
	\det(A_1)=\left(\mathcal I_0 \mathcal E_0\sqrt{V_E(\bu, \bv)}\right)^2\neq 0
\end{equation}
and, as a consequence, 
\begin{equation}
	\begin{aligned}
	&D \tilde G(\bar q)= \begin{pmatrix}
		B_u & 0\\ 0 & B_v
	\end{pmatrix}, \qquad
	 B_u=\begin{pmatrix}
	b_{11} & b_{12}\\ b_{21} & b_{22}
\end{pmatrix}\\
& b_{11}=\frac{\epsilon_1(\epsilon_1+\eta_1+\mathcal I_0)+\mathcal E_0(2\epsilon_1+\eta_1+\mathcal{I}_0)}{\mathcal{E}_0\mathcal{I}_0}\\
& b_{12}=\sqrt{V_E(\Gamma(\bu,\bv))}\frac{\mathcal E_0+\epsilon_1+\mathcal{I}_0+\eta_1}{\mathcal{E}_0\mathcal{I}_0}\\
& b_{21}=\frac{2\mathcal E_0 \epsilon_1(\mathcal I_0+\eta_1)+\epsilon_1^2(\mathcal{I}_0+\eta_1)+\mathcal{E}_0\eta_1(2\mathcal{I}_0+\eta_1)+\epsilon_1\eta_1(2\mathcal{I}_0+\eta_1)}{\mathcal{E}_0\mathcal{I}_0\sqrt{V_E(\Gamma(\bu,\bv))}}\\
& b_{22}=\frac{\mathcal{E}_0(\mathcal I_0+\eta_1)+\epsilon_1(\mathcal I_0+\eta_1)+\eta_1(2\mathcal I_0+\eta_1)}{\mathcal{E}_0\mathcal{I}_0}. 
\end{aligned}
\end{equation}
and analogous relations for $B_v,$ replacing $\epsilon_1, \eta_1$ with $\epsilon_2, \eta_2$. Again, it is easy to prove that $\det(D\tilde G(\bar q))=\det(B_u)=\det(B_v)=1,$ and that its four eigenvalues are coupled into two pairs $(\lambda_u,\lambda_u^{-1})$, $(\lambda_v, \lambda_v^{-1})$, which are either real or conjugate complex depending on the sign of 
\begin{equation}
	\begin{aligned}
		\Delta'_u=tr(B_u)^2-4=\frac{(\epsilon_1+\eta_1)}{\mathcal{E}_0^2\mathcal{I}_0^2}\left(2\mathcal E_0+\epsilon_1+\eta_1\right)\left(2\mathcal I_0+\epsilon_1+\eta_1\right)\left(2\mathcal E_0+\epsilon_1+2\mathcal I_0+\eta_1\right)\\
		\Delta'_v=tr(B_v)^2-4=\frac{(\epsilon_2+\eta_2)}{\mathcal{E}_0^2\mathcal{I}_0^2}\left(2\mathcal E_0+\epsilon_2+\eta_2\right)\left(2\mathcal I_0+\epsilon_2+\eta_2\right)\left(2\mathcal E_0+\epsilon_2+2\mathcal I_0+\eta_2\right)
	\end{aligned}. 
\end{equation} 
 According to the same rules already described in Section \ref{ssec:stab_Kep}, the signs of $\Delta'_u$ and $\Delta'_v$ determine if the homothetic point $\bar q$, in the four-dimensional first return map for the refractive dynamics, is of centre-centre, saddle-centre or saddle-saddle type,. Substituting $\mathcal I_0, \mathcal E_0, \epsilon_{1\backslash2}, \eta_{1\backslash2}$ with their formulas in terms of principal curvatures and physical parameters, given in Eqs. \eqref{eq:der2S} and \eqref{eq:E0eps}, one finds Theorem \ref{thm:stability_intro}, in the refractive case. \\
A comparison between the result obtained in the present paper and the one contained in \cite[Theorem 1.1]{deblasiterracinirefraction} shows that here we obtain the same result where we consider variations parallel to the principal directions of $S$ in $\bp$. 
	
	\section{Chaotic behaviour}\label{sec:chaos}

	The main scope of Section \ref{sec:stability} was the analysis of the first return map, in both  reflective and refractive models, \emph{locally} around a particular class of equilibrium trajectories defined by  central configurations. In the present section, central configurations will be the key ingredient to prove a more global property on the dynamics, related to the presence of \textit{chaotic behaviour}, at least for a subset of initial conditions. More precisely, we will link the geometric properties of the domain $D$ to the presence of topological chaos for large inner energies. 
	As already described in Section \ref{sec:intro}, the proving scheme is based on finding  a conjugation between the first return map, expressed in a suitable set of variables, and the Bernoulli shift of two symbols, constructing trajectories which are the composition of arcs that \textit{shadow} either homothetic segments or a juxtaposition of them. What we will obtain is a \emph{symbolic dynamics,} where billiard trajectories of the two models can be encoded by sequences of symbols (see \cite{Dev_book}).

	\subsection{Chaoticity of the reflective Kepler billiard}\label{ssec:chaor_refl}
	Let us start by providing the proof of Theorem \ref{thm:intro_chaos} in the case of reflection billiards, and then suppose that the domain $D$ admits two non antipodal and non degenerate central configurations, called $\bp_1$ and $\bp_2$. As in the previous sections, let us construct two local charts around $\bp_1, \bp_2$, given by $\Gamma^{(i)}: U_i\to S$, $U_i\subset \R^2$ open, $i=1,2$, with the properties described by Remark \ref{rem:par}, where we will denote with $\ku^{(i)}, \kd^{(i)}$ the principal curvatures of $S$ in $\bp_i$. 	\\		
	Given $\varepsilon>0, $ let us define the squares
	\begin{equation}\label{eq:square}
		R^{(i)}_\varepsilon\eqdef\left[\uu{i}-\varepsilon, \uu{i}+\varepsilon\right]\times\left[\vv{i}-\varepsilon, \vv{i}+\varepsilon\right], \quad i=1,2,  
	\end{equation}
	and the corresponding compact images on $S$ given by $V^{(i)}_\ve\eqdef \G{i}\left(R^{(i)}_\ve\right)$. 
	By Theorem \ref{thm:existence_inner}, if $\ve$ is sufficiently small and we fix $\H$ sufficiently large, for every $p_0, p_1\in \cup_{i=1}^2V^{(i)}_\ve$ we have a unique inner arc $z_I(\cdot; p_0, p_1)$, not homotopic to $\overline{p_0p_1}$ in $span(0, p_0,p_1)$, connecting them; in the following, we will always assume $\H>\overline{\H}$, such threshold being defined in Theorem \ref{thm:existence_inner}.  \\
	As a consequence, in the above setting, we can ensure the good definition and regularity of the inner generating function  
	\begin{equation}
		\S_I: \left(\bigcup_{i=1}^2R^{(i)}_\ve\right)^2\to \R. 
	\end{equation} 
	\begin{definition}\label{def:conc}
		Given a bi-infinite sequence of points $\underline p=(p_k)_{k\in\Z}$, $p_k\in \cup_{i=1}^2V^{(i)}_\ve$, let us call $\tilde z(\cdot; \underline p): \R\to \overline{D}$ the unique concatenation of inner arcs connecting  $p_{k}$ to $p_{k+1}$, $k\in\Z$. \\
		Equivalently, given the sequence $\bar \xi = (u_k, v_k)_{k\in\Z}$ such that $(u_k, v_k)\in\cup_{i=1}^2 R^{(i)}_\ve$, $p_k=\G{1\backslash 2}(u_k, v_k)$, the same concatenation can be denoted by $\tilde z(\cdot; \underline \xi)$.  \\
		Possibly translating the time, we will always suppose that $\tilde z(0; \un p)=p_0. $
	\end{definition}
	We notice that, whenever $\underline p$ is a $n$-periodic sequence with periodicity modulus $[\underline p]=(p_1, \dots, p_n)$, the concatenation $\tilde z(\cdot; \underline p)$ is itself a periodic trajectory composed by $n$ Keplerian arcs.
	\begin{definition}
		Given $\un p$ periodic with $\underline [\underline p]=(p_1, \dots,p_n)$, and $\bar \xi$ the corresponding sequence of parameters, $[\un \xi]=(u_1,v_1, \dots, u_n, v_n)$, the \emph{total Jacobi length} of the concatenation $\tilde z(\cdot; \un p)$ is given by 
		\begin{equation}\label{eq:periodic_length}
			W(\un p)=\sum_{i=1}^{n}\mathcal L_I(z_I(\cdot; p_i, p_{i+1}))=\sum_{i=1}^n\S_I(u_i, v_i, u_{i+1}, v_{i+1}),
		\end{equation}
	with the convention that $p_{n+1}=p_1$. Again, with an abuse of notation we will denote the above quantity as well with $W(\un \xi)$. 
	\end{definition}
	Let us now start with the construction of our symbolic dynamics, finding a way to describe the concatenations $\tilde z$ through bi-infinite sequences of two symbols, each of them corresponding to the neighbourhood of one of the two central configurations. To this aim, let us define $\mathbb L\eqdef\{1,2\}^{\mathbb Z}$, endowed with the distance (see \cite{hasselblattkatok}) 
	 \begin{equation}\label{eq:distance}
	 	\forall \un s=(s_k)_{k\in \Z},\  \un s'=(s'_k)_{k\in \Z}\in \mathbb L\quad d(\un s, \un s')\eqdef\sum_{k\in \Z}\frac{\delta(s_k, s'_k)}{4^{|k|}}, \quad \delta(i,j)=\begin{cases}
	 		1  \ &\text{if }i\neq j\\
	 	0  \ &\text{if }i= j
	 	\end{cases}. 
	 \end{equation}
	\begin{definition}\label{def:realises}
		Given $\un p$ as in Definition \ref{def:conc} and the corresponding concatenation $\tilde z(\cdot; \un p)$, we say that it \emph{realises} a sequence $\un s\in \mathbb L$ if and only if $\forall k\in \Z \quad p_k\in V^{(s_k)}_\ve$,	namely, in terms of parameters $\un \xi$, it holds $(u_k, v_k)\in R^{(s_k)}_\ve$. 
	\end{definition}
	The correspondence between trajectories and bi-infinite sequences in $\mathbb L$ is illustrated in Figure \ref{fig:realises}. 
	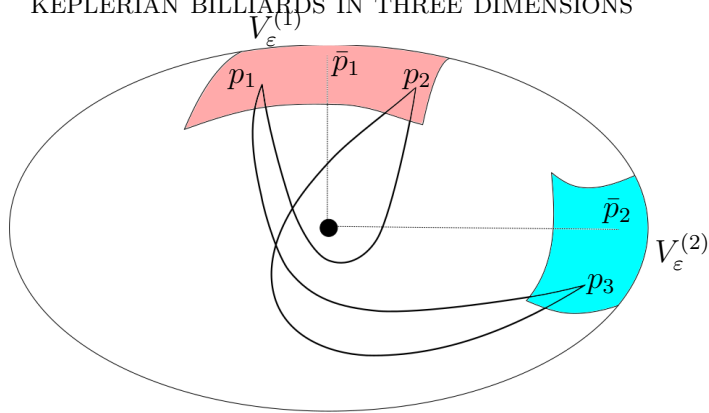
\begin{figure}
		\centering
		\begin{overpic}[width=0.5\linewidth]{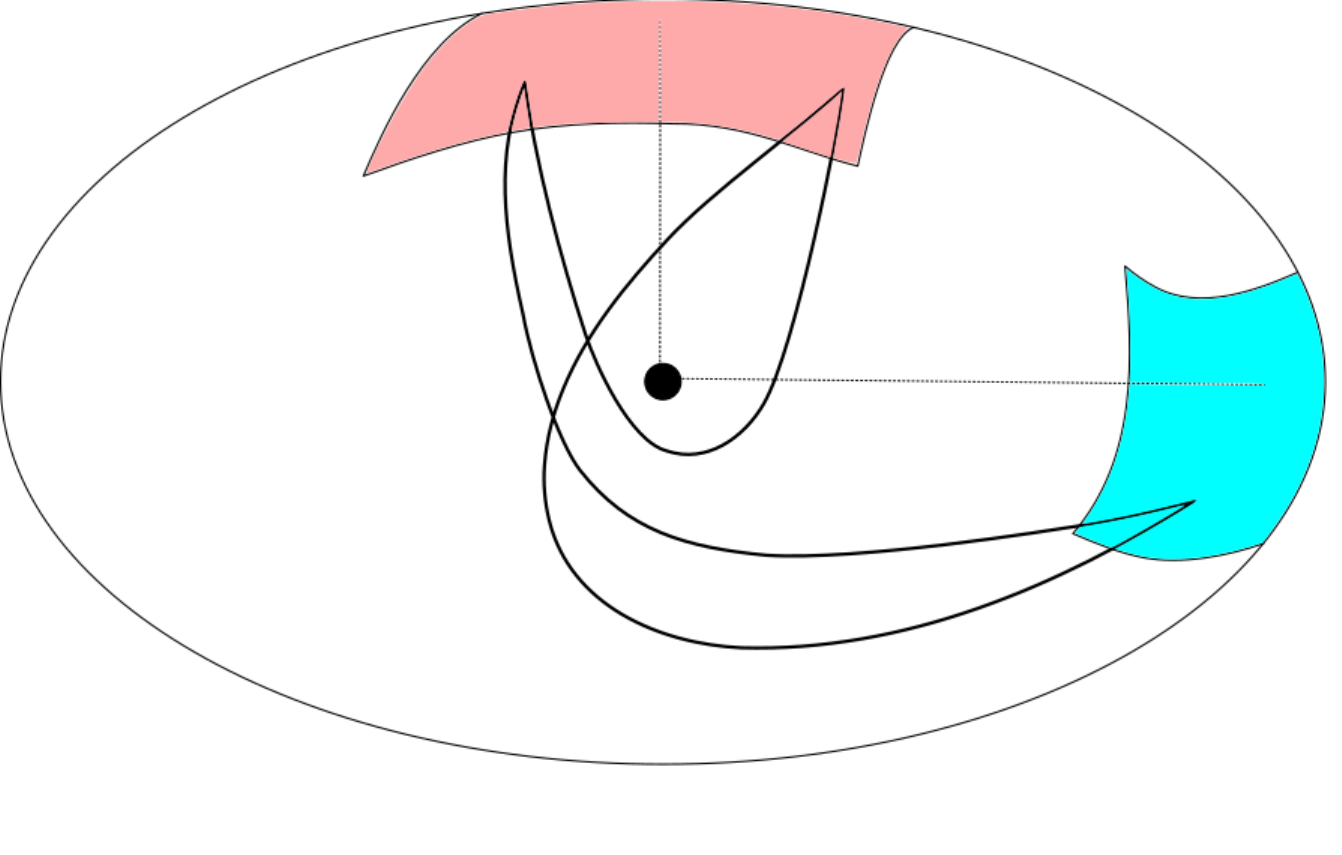}
			\put (37, 65) {\rotatebox{0}{$V^{(1)}_\ve$}}
			\put (40, 96) {\rotatebox{0}{$w_0$}}
			\put (92,37) {\rotatebox{0}{$\bp_2$}}
			\put (61,58) {$p_2$}
			\put (50, 60) {$\bp_1$}
			\put (34, 58) {$p_1$}
			\put (89.5, 26) {$p_3$}
			\put (100, 30) {$V^{(2)}_\ve$}
		\end{overpic}
		\caption{Illustration of a periodic trajectory realising the sequence $\un s\in\mathbb L$ with periodicity modulus $[s]=(1,1,2)$. It is composed by three Keplerian arcs, hitting the boundary respectively in $V^{(1)}_\ve$ $V^{(1)}_\ve$ and $V^{(2)}_\ve$. }
		\label{fig:realises}
	\end{figure}
	Although for any $\un p\in \otimes_{\Z}\left(\cup_{i=1}^2V^{(i)}_\ve\right)$ there exists a unique concatenation connecting them, it is clear that a priori nothing ensures that such concatenation satisfies the reflection law at every bouncing point: the next proposition guarantees that, at least in the periodic case, there is a one-to-tone correspondence between sequences in $\mathbb L$ and billiard trajectories realising them. 
	\begin{proposition}\label{prop:periodic}
		Possibly restricting $\ve$, there exists $\overline{\H}_1$ such that, for every  $\H>\overline{ \H}_1$ and $\un s \in \mathbb L$ periodic, exists a unique billiard trajectory that realises $\un s$.  
	\end{proposition} 
	\begin{proof}
		Let us fix $\un s \in \mathbb L$ periodic, with $[\un s]=(s_1, \dots, s_n)$ and let us define the set 
		\begin{equation}
			\mathcal U_{\un s}\eqdef \otimes_{k=1}^nR_\ve^{(s_k)}; 
		\end{equation}
	identifying every $\un \xi$ with its periodic extension, we already know that, for every $\un \xi\in \mathcal U_{\un s}$, there exists a unique periodic concatenation  $\tilde z(\cdot; \un \xi)$ which connects $\left(\Gamma^{(s_k)}(u_k, v_k)\right)_{k\in\Z}$, whose length is given by \eqref{eq:periodic_length}. Moreover, by the variational characterization of reflected arcs depicted in \eqref{eq:riflS}, we can notice that, given $\hat{\un \xi}\in \mathcal U_{\un s},$ the corresponding  $\tilde z(\cdot; \hat{\un \xi})$ satisfies the reflection law in every bouncing point if and only if $\nabla W(\hat{\un\xi})=\underline 0$. In practice, the problem of finding a billiard trajectory realising $\un s$ is translated into searching for a critical point of the length functional $W{(\un \xi)}$ in the set $\mathring{\mathcal U}_{\un s}$.  \\
	To find such critical point, we shall make use of the	\emph{Poincaré-Miranda Theorem} (see \cite{Miranda} and \cite[Theorem 2.12]{deblasiterracinibarutelloChaotic}). In particular, if we are able to prove that, for every $k=1, \dots, n$, 
	\begin{equation}\label{eq:change_sign}
		\begin{cases}
			\partial_{u_k}W(\un \xi)_{|u_k=\uu{s_k}-\ve} \cdot \partial_{u_k}W(\un \xi)_{|u_k=\uu{s_k}+\ve} < 0\\
			\partial_{v_k}W(\un \xi)_{|v_k=\vv{s_k}-\ve} \cdot \partial_{v_k}W(\un \xi)_{|v_k=\vv{s_k}+\ve} < 0, 
		\end{cases}
	\end{equation} 
	then there is a critical point of $W_{\un s}$ in $\mathring{\mathcal U}_{\un s}$. To prove \eqref{eq:change_sign}, let us introduce an asymptotic relation holding for fixed-ends Keplerian arcs when $\H\to\infty$ (see \cite{bolotin2000periodic} and \cite[Appendix B]{deblasiterracinibarutelloChaotic}): for every $p_0, p_1\in \R^3\setminus\{0\}$ not antipodal and every $\H>0$
	\begin{equation}
		d_I(p_0, p_1)=\sqrt{\H}\left(|p_0|+|p_1|\right)+\frac{\mu}{\sqrt{\H}}\left(g(p_0, p_1; \H)-\log\left(	\frac{\mu}{2\H}\right)\right), 
	\end{equation} 
	where $g(p_0, p_1; \H)$ is $C^2-$bounded when $\H\to\infty$ uniformly in $p_0, p_1$. \\
	As a consequence, for every $k=1, \dots, n$ we have 
	\begin{equation}\label{eq:asintotica}
			\begin{aligned}
			S(u_k,v_k, u_{k+1}, v_{k+1})=&\sqrt{\H}(|\Gamma^{(s_k)}(u_k, v_k)|+|\Gamma^{(s_{k+1})}(u_{k+1}, v_{k+1})|)\\&+\frac{\mu}{\sqrt{\H}}\left(g\left(\Gamma^{(s_k)}(u_k, v_k),\Gamma^{(s_{k+1})}(u_{k+1}, v_{k+1}); \H\right)-\log\left(\frac{\mu}{2 \H}\right)\right),
		\end{aligned}
	\end{equation}
	\begin{equation}\label{eq:Fkasintotiche}
		\begin{aligned}
			\partial_{u_k}W(\un \xi)=&2\sqrt{\H}\partial_{u_k}|\Gamma^{(s_k)}(u_k,v_k)|
				+\frac{\mu}{\sqrt{\H}}\Big(\partial_{2,u}g\left(\Gamma^{(s_{k-1})}(u_{k-1}, v_{k-1}),\Gamma^{(s_{k})}(u_{k}, v_{k}); \H\right)\\&+\partial_{1,u}g\left(\Gamma^{(s_k)}(u_k, v_k),\Gamma^{(s_{k+1})}(u_{k+1}, v_{k+1}); \H\right)\Big)\\
			\partial_{v_k}W(\underline\xi)=&2\sqrt{\H}\partial_{v_k}|\Gamma^{(s_k)}(u_k,v_k)|
			+\frac{\mu}{\sqrt{\H}}\Big(\partial_{2,v}g\left(\Gamma^{(s_{k-1})}(u_{k-1}, v_{k-1}),\Gamma^{(s_{k})}(u_{k}, v_{k}); \H\right)\\&+\partial_{1,v}g\left(\Gamma^{(s_k)}(u_k, v_k),\Gamma^{(s_{k+1})}(u_{k+1}, v_{k+1}); \H\right)\Big). 
		\end{aligned}
	\end{equation}
	By the boundedness of $g$, if $\H$ is sufficiently large,  the sign of the above quantities is completely determined by the partial derivatives of $|\Gamma^{(s_k)}(u_k,v_k)|$. 
	Recalling that all the pairs $(u_k, v_k)$ are close to the ones defining central configurations $(\uu{i}, \vv{i})$, the following Lemma, which will be proved rigorously afterwards, holds. 
	\begin{lemma}\label{lem:change_sign}
		If $\bp \in S$ is a non degenerate central configuration and the local chart around $\bp$ is parameterised as in Remark \ref{rem:par}, then there exists $\bar\ve>0$ such that, for every $0<\ve<\bar\ve$, $\exists$ $C(\ve)>0$ such that $R_\ve\eqdef[\bu-\ve, \bu+\ve]\times[\bv-\ve,\bv+\ve]\subset U$ and for every $(u_0, v_0)\in R_\ve$ 
		\begin{equation}\label{eq:change_euc}
			\partial_u|\Gamma(\bu-\ve, v_0)|\cdot \partial_u|\Gamma_u(\bu+\ve, v_0)|<-C(\ve), \quad \partial_v|\Gamma_v(u_0, \bv-\ve)|\cdot\partial_v|\Gamma_v(u_0, \bv+\ve)|<-C(\ve). 
		\end{equation}
	\end{lemma}
	Let us then take $\ve, C$ such that the above Lemma holds for both $\bp_1$ and $\bp_2$: taking $\H$ sufficiently large, Eqs. \eqref{eq:change_sign} are true for every $k=1, \dots, n$, and then, by Poincaré-Miranda theorem, there exists $\hat{\un\xi}\in \mathring{\mathcal U}_{\un s}$ such that $\nabla W(\hat{\un \xi})=\un 0$. As a consequence, $\tilde z(\cdot; \hat{\un\xi})$ is a periodic billiard trajectory which realises $\un s$. \\
	As for the uniqueness of the critical point $\hat{\un\xi}$ (and, as a consequence, of the corresponding periodic billiard trajectory), let us notice that, since $\bp_1$ and $\bp_2$ are non degenerate, by Eqs.  \eqref{eq:der_par}, $\partial^2_{u\backslash v}|\G{i}(\uu{i}, \vv{i})|$ are different from 0. As a consequence, taking a sufficiently small $\ve$ and $\H$ sufficiently large, one has that for every $k=1, \dots, n$ the quantities $\partial^2_{u_k}W(\un\xi)$ and $\partial^2_{v_k}W(\un\xi)$ are bounded from zero: this implies the strict monotonicity of $\partial_{u_k}W$, $\partial_{v_k}W$ for every $k=1, \dots, n$, and therefore the uniqueness of  $\hat{\un\xi}$.  \\
	Let us conclude by noticing that, for every $\un s\in \mathbb L$ periodic, the partial derivatives of $W(\un \xi): \mathcal U_{\un s}\to \R$ always have the same form, involving only three consecutive pairs $(u,v)$ of $\un \xi$: this means that the estimates used above are independent on $\un s$, and that the thresholds found for $\ve$ and $\H$ are uniform in the periodic words of $\mathbb L. $ 
	\end{proof}

\begin{proof}[Proof of Lemma \ref{lem:change_sign}]
	To fix the ideas, let us suppose that $(\bu, \bv)$ is a saddle point for $|\Gamma(u, v)|$: recalling Eqs. \eqref{eq:der_par}, 
	\begin{equation}
		A_1\eqdef\left(\frac{1}{|\bar p|}-\kappa_1\right)=\partial_u^2|\Gamma(\bu,\bv)|<0,\ A_2\eqdef\left(\frac{1}{|\bar p|}-\kappa_2\right)\partial_v^2|\Gamma(\bu,\bv)|>0.   
	\end{equation}
By the regularity of $S$, for every $(u_0, v_0)\in U$ there are $(u',v'), (u'',v'')$ in the segment between $(\bu,\bv)$ and $(u_0,v_0)$ such that 
	\begin{equation}
		\begin{aligned}
			\partial_ u|\Gamma(u_0,v_0)|=&\partial_u|\Gamma(\bu,\bv)|+\partial^2_u|\Gamma(\bu,\bv)|(u_0-\bu)+\partial_{uv}|\Gamma(\bu,\bv)|(v_0-\bv)\\
			&+ \frac{\partial^3_{u}|\Gamma(u',v')|}{2}(u_0-\bu)^2+\partial_{uuv}|\Gamma(u',v')|(u_0-\bu)(v_0-\bv)\\
			&+\frac{\partial_{uvv}|\Gamma(u',v')|}{2}(v_0-\bv)^2\\
			\partial_ v|\Gamma(u_0,v_0)|=&\partial_v|\Gamma(\bu,\bv)|+\partial^2_v|\Gamma(\bu,\bv)|(v_0-\bv)+\partial_{uv}|\Gamma(\bu,\bv)|(u_0-\bu)\\
			&+ \frac{\partial^3_{v}|\Gamma(u'',v'')|}{2}(v_0-\bv)^2+\partial_{vvu}|\Gamma(u'',v'')|(u_0-\bu)(v_0-\bv)\\
			&+\frac{\partial_{vuu}|\Gamma(u'',v'')|}{2}(u_0-\bu)^2. 
		\end{aligned}
	\end{equation} 
	Since $\Gamma$ is at least $C^3$, the third derivatives of $|\Gamma(u,v)|$ are uniformly bounded in $(u,v)$: let us call $C_1>0$ the constant bounding their moduli.  \\
	For every $\ve>0$ such that $R_\ve\eqdef[\bu-\ve, \bu+\ve]\times[\bv-\ve, \bv+\ve]\subset U$ and every $(u_0, v_0)\in R_\ve$, one has
	\begin{equation}
		\begin{aligned}
			&\partial_u|\Gamma(\bu-\ve,v_0)|&&=-A_1\ve+ \frac{\partial^3_{u}|\Gamma(u',v')|}{2}\ve^2-\partial_{uuv}|\Gamma(u',v')|\ve(v_0-\bv)\\&\quad&&+\frac{\partial_{uvv}|\Gamma(u',v')|}{2}(v_0-\bv)^2\\
			&\ &&\geq -A_1\ve -2 C_1 \ve^2\\
			&\partial_u|\Gamma(\bu-\ve,v_0)|&&\leq A_1 \ve +2C_1\ve^2, \quad 
			\partial_v|\Gamma(u_0,\bv-\ve)|\leq -A_2 \ve +2C_1\ve^2, \\ 
			&\partial_v|\Gamma(u_0,\bv+\ve)|&&\geq A_1 \ve -2C_1\ve^2, 
		\end{aligned}
	\end{equation}
	and then, taking $\ve$ sufficiently small, one can find Eqs. \eqref{eq:change_euc}.
\end{proof}	
	Starting from the one-to one correspondence between sequences in $\mathbb L$ and periodic billiard trajectories realizing them, it is possible to construct a symbolic dynamics in the case of reflective three dimensional billiards. From this moment on, we will always assume that  $\ve$ and $\H$ satisfy the hypotheses of Proposition \ref{prop:periodic}. \\
	Following what already described in Section \ref{sec:intro}, let us start by defining the projection map $P$. 
	\begin{definition}\label{def:proiezione}
		Let $\un s\in\mathbb L$, and suppose that there exists $\tilde z(\cdot, \un p)$ that satisfies $\un s$, according to Definition \ref{def:realises}. Setting $(p_0,w_0)\eqdef\left(\tilde z(0; \un p), \tilde z'(0; \un p)\right)$, we say that $P(p_0, w_0)=\un s$. 
	\end{definition}
	In practice, the projection map links initial conditions of billiard trajectories to sequences in $\un s$ realised by them. Let us observe that, whenever $\tilde z$ is a billiard trajectory and not a simple concatenation of arcs, it is completely determined by the initial conditions. \\
	From Proposition \ref{prop:periodic}, we already know that if $\un s$ is periodic there is a unique pair of initial conditions $(p_0, w_0)\in V_\ve^{(s_0)}\times \R^3$ which is projected onto $\un s$, then $P$ is well defined in this case. We will now construct a subset $X$ of the initial conditions for the first return map $F: (p_0, w_0)\mapsto (p_1, w_1)$ as defined in Section \ref{sec:intro}, such that:  
	\begin{itemize}
		\item $F$ is infinitely-many times forward and backward well defined on $X$, and $F(X)=X$; 
		\item the projection $P$ is well defined on $X$, that is, all orbits with initial conditions in $S$ intersect $S$ only in $\cup_{i=1}^2V^{(i)}_\ve$; 
		\item  recalling the definition of Bernoulli shift given in \eqref{eq:Bernoulli}, it holds $\sigma \circ P= P\circ F$. 
	\end{itemize}
	If the above conditions are satisfied, it is possible to construct the commutative diagram 
	\begin{equation}\label{eq:diagram}
	\begin{tikzcd}
		X \arrow{r}{{F}} \arrow{d}{P} & X \arrow{d}{P} \\
		\mathbb{L} \arrow{r}{\sigma}	& \mathbb{L}
	\end{tikzcd}
	\end{equation}
	From \cite{Dev_book}, if $P$ is bijective and continuous, then $F$ is topologically chaotic on $X$, and then Theorem \ref{thm:intro_chaos} holds in the reflective case. \\
	Let us start with the construction of $X$, and consider the set of initial conditions 
	\begin{equation}
		\Psi=\left\{(p_0, w_0)\in S\times \R^3 \ |\  p_0\in \bigcup_{i=1}^2V^{(i)}_\ve, \ |w_0|=\sqrt{2 V_I(p_0)}, \ w_0 \text{ points inside }D\right\}. 
	\end{equation}
	 Any initial condition in $\Psi$ leads to a Keplerian inner arc which intersects again $S$. Let us restrict $\Psi$ so that the corresponding arcs encounter only $V_\ve^{(i)}$, for some $i=1,2$, and with velocities transversal to $S$: let us call this set $A\subset \Psi$.  We can ensure that $F: (p_0, w_0)\mapsto (p_1, w_1)$ (see also Figure \ref{fig:primo ritorno Kep}) is well defined in $A$; for the reversibility of the flow associated to the inner problem, we can also observe that the inverse mapping $F^{-1}$ is well defined on $F(A)$.
	Let us now consider the set
	\begin{equation}
		X\eqdef\bigcap_{k\in\Z}F^k(A), 
	\end{equation} 
	that is, the set of initial conditions for which all the backward and forward iterates of $F$ are well defined and the corresponding billiards trajectories intersect $S$ in $\cup_{i=1}^2V^{(i)}_\ve$. \\
	By virtue of Proposition \ref{prop:periodic} $X$ is not empty, and it is invariant under $F$: with an abuse of notation, we will still denote with $F$ the restriction of the first return map to $X$. \\
	Let us now observe that the projection map $P$ is surely well defined on $X$: given $(p_0, w_0)\in X$, let us call $\un p=(p_k)_{k\in Z}$ the sequence of all the bouncing points of the corresponding billiard trajectory; we can then say that 
	\begin{equation}\label{eq:proj}
		P(p_0, v_0)=\un s \in \mathbb L\quad \Longleftrightarrow\quad  \forall k\in \Z \ \quad p_k\in V^{(s_k)}_\ve. 
	\end{equation}
	\begin{proposition}\label{prop:conj_rifl}
		Given $P:X\to \mathbb L$ defined as in \eqref{eq:proj}, it holds that \begin{enumerate}
			\item the diagram in \eqref{eq:diagram} commutes; 
			\item $P$ is bijective; 
			\item $P$ is continuous. 
		\end{enumerate}
	\end{proposition}
\begin{proof}
	Property (1) is straightforward by construction: if we consider $(p_0, w_0), (p_1, w_1)\in X$ such that $(p_1, w_1)=F(p_0, w_0)$, it is clear that the sequences of bouncing points starting from $p_0$ and $p_1$ are one the shift of the other. \\
		As for the second point, from Proposition \ref{prop:periodic} we can say that whenever $\un s $ is periodic there exists a unique $(p_0, w_0)\in X$ which projects into $\un s$. Let us then take $\un s\in \mathbb L$ a non periodic sequence, 
		and consider, for every $n\geq 0$, $\un s^{(n)}$ periodic with $[\un s^{(n)}]=(s_{-n}, \dots, s_0, \dots, s_n)$; considering the distance defined in \eqref{eq:distance}, one has $\un s^{(n)}\to \un s$ as $n\to\infty$. 
		Again by proposition \ref{prop:periodic} for every $n\geq 0$ there exists a unique ${\un\xi}^{(n)}\in \mathcal U_{\un{s}^{(n)}}$ such that the corresponding concatenation $\tilde z(\cdot; {\un\xi}^{(n)})$ realises $\un s^{(n)}$. With an abuse of notation, let us extend $\un\xi^{(n)}$ by periodicity, and denote  $[{\un\xi}^{(n)}]=(u_{-n}^{(n)},v_{-n}^{(n)}, \dots, u_{0}^{(n)},v_{0}^{(n)}, \dots,  u_{n}^{(n)},v_{n}^{(n)})$. For every fixed $k\in \Z$, it results that $(u_k^{(n)}, v_k^{(n)})\in R_\ve^{(s_k)}$ for $n$ sufficiently large: by compactness, $(u_k^{(n)}, v_k^{(n)})\to (\tilde u_k, \tilde v_k)$ for some $(\tilde u_k, \tilde v_k)\in R^{(s_k)}_\ve$ up to a subsequence. \\
		With a diagonalization procedure, we can state the existence of a subsequence $(a_n)_{n\in \N}\subset \N$ and a sequence $\tilde{\un \xi}\in \otimes_{k\in\Z}R_\ve^{(s_k)}$ such that 
		\begin{equation}
			\forall k\in\Z\quad \left(u_k^{(a_n)}, v_k^{(a_n)}\right)\to (\tilde u_k, \tilde v_k). 
		\end{equation}
	Let us then consider $z(\cdot; \tilde{\un\xi})$  the concatenation of inner arcs which connects the points $\tilde p_k=\Gamma^{(s_k)}(\tilde u_k, \tilde v_k)$: by convergence and the $C^1$ dependence of both Cauchy problems and the reflection law from initial conditions, $z(\cdot; \tilde{\un\xi})$ is a billiard trajectory that realises $\un s$. The uniqueness of $z(\cdot; \tilde{\un\xi})$, and, as a consequence, of the corresponding initial conditions, follows from the uniqueness of the sequences ${\un\xi}^{(n)}$ in $\mathcal U_{\un s^{(n)}}$. \\
	Let us now pass to the continuity of $P$, and start with a preliminary observation. Taking any two points $p_0, p_1\in \cup_{i=1}^2V^{(i)}_\ve$, it holds that the time that $z_I(\cdot; p_0, p_1)$ takes to go from $p_0$ to $p_1$ is bounded uniformly in the endpoints. \\
	Let us now fix $(p_0, w_0)\in X$, and call $\un s$ the sequence in $\mathbb L$ such that $P(p_0, w_0)=\un s$. By the convergence of $d(\cdot, \cdot)$ in $\mathbb L$, for every $\epsilon>0$ there exists $k_0\in \N$ such that, for every $\un s' \in \mathbb L$, 
	\begin{equation}
		\sum_{|k|> k_0}\frac{\delta(s_k, s'k)}{4^{|k|}}<\epsilon; 
	\end{equation}
	if we will prove that there exists $\delta>0$ such that, if $|(p_0, w_0)-(p_1, w_1)|<\delta$, then 
	\begin{equation}\label{eq:sum0}
		\sum_{|k|\leq k_0}\frac{\delta(s_k, s'_k)}{4^{|k|}}=0, 
	\end{equation}
	then the continuity of $P$ over $X$ follows immediately. Eq. \eqref{eq:sum0} corresponds to ask that, called $\un s'=P(p_1, v_1)$, for every $k=-k_0, \dots, 0, \dots, k_0$, we have the equality $s_k=s'_k$, namely, the billiard trajectories with initial conditions $(p_0, w_0)$ and $(p_1,w_1)$  bounce in the same compact sets of $S$ for $2k_0+1$ times. By the boundedness of the times to go from two points with an inner arc, we can find a time interval $[-a, a]$ such that every billiard trajectory with initial conditions in $X$ encounters the boundary $S$ at leats $2k_0+1$ times, and then the thesis follows from continuity of Cauchy problems with respect to initial conditions. 	
\end{proof}
	Proposition \ref{prop:conj_rifl} completes the proof of Theorem \ref{thm:intro_chaos} in the reflective case, defining a conjugation between the Bernoulli shift and the first return map $F$ restricted to the subset of initial conditions $X$. \\
	Let us conclude this section by highlighting a geometrical property of the billiard trajectories starting with initial conditions in $X$: not only they cross $S$ close to the homothetic trajectories defined by $\bp_1, \bp_2$; by continuity, they remain \emph{close} to them for all forward and backward times. For this reason, one can refer to this type of construction as a \emph{shadowing} procedure, since defining $P$ one ends up in building trajectories close to juxtaposition of homothetic arcs (see Figure \ref{fig:realises}). As we will see in the next section, an analogous scheme will be applied in the refractive case. 

	\subsection{Chaoticity of the three dimensional refractive galactic billiard}\label{ssec:chaos_refr}
	
	The aim of this section is to apply the reasoning already illustrated in Section \ref{ssec:chaor_refl} to the refractive model: although in principle the procedure followed is the same as in the reflective case, the presence of the outer dynamics, and the fact that the first return map follows a concatenation of two arcs instead of one at a time, doubles the dimensionality of the problem. Also in this case, we will conjugate the Bernoulli shift of two symbols with the refractive first return map $G$ as defined in Section \ref{sec:intro} and Figure \ref{fig:primo_rit_snell}, restricting it to a suitable subset of initial condition of trajectories which \emph{shadow} either homothetic trajectories or juxtapositions of them; for the sake of brevity, we will not re-write the proofs which are analogous of the ones in Section \ref{ssec:chaor_refl}, limiting ourselves in highlighting the differences between the two models.  \\
	Let us then consider $\bp_1$ and $\bp_2$ as in Section \ref{ssec:chaor_refl}, and suppose again that the local charts around them have the properties listed in Remark \ref{rem:par}. Defining the squares  $R^{(i)}_\ve$ as in \eqref{eq:square}, and the corresponding images on $S$, denoted by $V^{(i)}_\ve$, Theorems \ref{thm:existence_inner} and \ref{thm:outer problem}  ensure that,  whenever $\ve$ is sufficiently small and $\H$ large enough, 
	\begin{equation}
		\begin{aligned}
		&	\forall p_0, p_1\in V^{(i)}_\ve, \ i=1,2,\ \exists|z_E(\cdot; p_0, p_1)\text{ outer arc connecting }p_0\ \text{to } p_1; \\
		&	\forall p_0, p_1\in \bigcup_{i=1}^2V^{(i)}_\ve, \exists|z_I(\cdot; p_0, p_1)\text{ inner arc as in Theorem \ref{thm:existence_inner} connecting }p_0\ \text{to } p_1.  
		\end{aligned}
	\end{equation}
	In terms of local coordinates, this translates into the good definition of the outer and inner generating functions (see Eq. \eqref{eq:gen})
	\begin{equation}
		\S_E: \bigcup_{i=1}^2\left(R^{(i)}_\ve\right)^2\to \R,\quad \S_I: \left(\bigcup_{i=1}^2R^{(i)}_\ve\right)^2\to \R. 
	\end{equation}
	As in the reflective case, we will now define a infinite concatenation of outer and inner arcs crossing the surface $S$ in $V^{(i)}_\ve$, and, in the case of periodic trajectories, the relative Jacobi length. In this case, we will always suppose that the outer arcs connect points close to each others, while inner ones act as transfer trajectories, possibly moving the particle between different regions of $S$. To fix the notation, let us give the following definition. 
	\begin{definition}
		Let us consider a sequence of points $\un p=\left(p_k^{(E)}, p_k^{(I)}\right)_{k\in \Z}\subset\left(V^{(i)}_\ve\times V^{(i)}_\ve\right)^\Z$, and consider the unique concatenation $\tilde z(\cdot; \un p):\R\to\R^3$ composed, for every $k\in \Z$, by the outer arc connecting $p_k^{(E)}$ to $p_k^{(I)}$ and the inner arc connecting $p_k^{(I)}$ to $p^{(E)}_{k+1}$. As in Definition \ref{def:conc}, we will always suppose that $\tilde z(0; \un p)=p_0^{(E)}$.   \\
		In terms of local parameterisation, let us set $\un \xi=(u_k^{(E)}, v_k^{(E)},u_k^{(I)}, v_k^{(I)})_{k\in\Z}$ the sequence of parameters corresponding to $\un p$; identifying $\tilde z(\cdot; \un \xi)\eqdef\tilde z(\cdot;\un p). $ Whenever $\un \xi$ (and hence $\un p$) is periodic with periodicity modulus 
		\begin{equation}[\un\xi]=\left(u_1^{(E)}, v_1^{(E)},u_1^{(I)}, v_1^{(I)} \dots,u_n^{(E)}, v_n^{(E)},u_n^{(I)}, v_n^{(I)}\right),
		\end{equation} the corresponding concatenation is composed by $2n$ arcs and its Jacobi length is given by 
		\begin{equation}
			W(\un \xi)=\sum_{k=1}^n \S_E(u_k^{(E)},v_k^{(E)}, u_k^{(I)},v_k^{(I)})+ \S_I (u_k^{(I)},v_k^{(I)}, u_{k+1}^{(E)},v_{k+1}^{(E)}),  
		\end{equation}
	where again $n+1=1$. 
	By virtue of Eq.\eqref{eq:refrS}, the concatenation $\tilde z(\cdot; \hat{\un \xi})$ is a billiard trajectory (i.e. satisfies Snell's law at every transition point) if and only if $\nabla W(\hat{\un\xi})=\un 0.$
		\end{definition}
	To construct a symbolic dynamics, let us take again the set of bi-infinite sequences $\mathbb L=\{1,2\}^\Z$ endowed with the distance \eqref{eq:distance}: given a concatenation of arcs $\tilde z(\cdot; \un p)$, we will say that $\tilde z$ \emph{realises} a sequence $\un s \in \mathbb L$ if and only if for every $k\in \Z$ it holds $p_k^{(E\backslash I)}\in V^{(s_k)}_\ve$. \\
	We are now ready to state the one-to-one relation between billiard trajectories and periodic sequences of $\mathbb L$ in the refractive case as well. 
	\begin{proposition}\label{prop:periodic_refr}
		Possibly restricting $\ve$, there exists $\overline \H_2>0$ such that, for every $\H>\overline \H_2$ and every $\un s\in\mathbb L$ periodic, there exists a unique billiard trajectory for the refractive case which realises $\un s$. 
	\end{proposition}
\begin{proof}
	The proof is analogous to the one of Proposition \ref{prop:periodic}, although the presence of an outer dynamics makes the estimates on $W$ more delicate. Again, the idea is to find a unique critical point for $W(\un\xi)$ inside a suitable set of parameters, provided $\ve$ is sufficiently small and $\H$ sufficiently large. For the sake of brevity, with an abuse of notation we will use the same symbols as in the previous section. 	\\First of all, let us notice that, by the regularity of $V_E$ and the compactness of $V^{(i)}_\ve$, $i=1,2$, once the outer energy $\mathcal E$ is fixed, one has that  $\S_E$, along with all its derivatives,  is bounded by a constant $C$ uniformly in the endpoints.  \\
	Let us then fix $\un s\in \mathbb L$ a periodic sequence, $[\un s]=(s_1, \dots, s_n)$, and consider 
	\begin{equation}
			\mathcal U_{\un s}\eqdef \otimes_{k=1}^n\left(R_\ve^{(s_k)}\right)^2: 
	\end{equation}
	given now $W: \mathcal U_{\un s}\to R$, by Miranda theorem we can ensure the existence of a critical point for the length function if we prove
	\begin{equation}\label{eq:ch_sign}
		\begin{cases}
			\partial_{u_k^{(E)}}W(\un \xi)_{|u_k=\uu{s_k}+\ve} \cdot \partial_{u_k^{(E)}}W(\un \xi)_{|u_k=\uu{s_k}+\ve} < 0\\
			\partial_{v_k^{(E)}}W(\un \xi)_{|v_k=\vv{s_k}+\ve} \cdot \partial_{v_k^{(E)}}W(\un \xi)_{|v_k=\vv{s_k}+\ve} < 0\\
			\partial_{u_k^{(I)}}W(\un \xi)_{|u_k=\uu{s_k}+\ve} \cdot \partial_{u_k^{(I)}}W(\un \xi)_{|u_k=\uu{s_k}+\ve} < 0\\
			\partial_{v_k^{(I)}}W(\un \xi)_{|v_k=\vv{s_k}+\ve} \cdot \partial_{v_k^{(I)}}W(\un \xi)_{|v_k=\vv{s_k}+\ve} < 0. 
		\end{cases}
	\end{equation}
	By direct computations and taking into account Eq. \eqref{eq:asintotica}, one has 
	\begin{equation}
		\begin{aligned}
			\partial_{u_k^{(E)}}W(\un \xi)&=\partial_{2,u}\S_I(u_{k-1}^{(I)},v_{k-1}^{(I)}, u_{k}^{(E)},v_{k}^{(E)})+ \partial_{1,u}\S_E(u_k^{(E)},v_k^{(E)}, u_k^{(I)},v_k^{(I)})\\
		\quad & = 2\sqrt{\H}\partial_{u_k}|\Gamma^{(s_k)}(u_k^{(E)},v_k^{(E)})|
		\\&+\frac{\mu}{\sqrt{\H}}\Big(\partial_{2,u}g\left(\Gamma^{(s_{k-1})}(u_{k-1}^{(I)}, v_{k-1}^{(I)}),\Gamma^{(s_{k})}(u_{k}^{(E)}, v_{k}^{(E)})\right)\Big)\\
		\quad & +\partial_{1,u}\S_E(u_k^{(E)},v_k^{(E)}, u_k^{(I)},v_k^{(I)}), 
		\end{aligned}
	\end{equation}
	and analogous formulas for all the other derivatives. By the uniform boundedness of the derivatives of $\S_E$ and the non-degeneracy of  $\bp_i$, one can use the same argument as in Proposition \ref{prop:periodic} to find a lower bound in $\H$, uniform in the sequences $\un s$, such that inequalities \eqref{eq:ch_sign} hold. One can then find $\hat{\un\xi}\in\mathring{\mathcal U}_{\un s}$ critical point for $W(\un\xi)$, whose uniqueness is again guaranteed possibly reducing $\ve$ starting from the non degeneracy of the central configurations. 
\end{proof}
	Proposition \ref{prop:periodic_refr} is again the starting point to construct a conjugation between the refractive first return map $G: (p_0, w_0)\mapsto (p_1, w_1)$, restricted to a suitable invariant set of initial conditions, and Bernoulli shift $\sigma$. The construction of the projection map $P$ from initial conditions of billiard trajectories and orbits is analogous to the one presented in Definition \ref{def:proiezione}: if $\tilde z(\cdot)\equiv\tilde z(\cdot; \bar p)$ realises a word $\un s\in \mathbb L$, then $P\left(\tilde z(0), \tilde z'(0)\right)=\un s$. \\
	Let us then construct a set of initial conditions $Y$ where we can ensure the good definition and invariance under $G$, along with the good definition of $P: Y\to \mathbb L$ and the commutativity of the diagram  
	\begin{equation}\label{eq:diagram_refr}
		\begin{tikzcd}
			Y \arrow{r}{{G}} \arrow{d}{P} & Y \arrow{d}{P} \\
			\mathbb{L} \arrow{r}{\sigma}	& \mathbb{L}
		\end{tikzcd}
	\end{equation}
	We refer to \cite{deblasiterracinibarutelloChaotic} for a more detailed explanation, holding in two dimensions. 
	Let us start by defining, for $i=1,2$, the sets 
	\begin{equation}
		\Xi_i^{(+)}=\left\{(p_0, w_0)\in V^{(i)}_\ve\times \R^3, \ | \ |w_0|=\sqrt{2V_E(p_0)}, \ w_0 \ \text{points outside $D$}\right\}; 
	\end{equation}
	in practice, we will search for initial conditions in $\cup_{i=1}^2\Xi_i^{+}$ which generates billiard trajectories satisfying sequences in $\mathbb L$. Recalling the definition of $G$ given in Section \ref{sec:intro}, and in particular Figure \ref{fig:primo_rit_snell}, let us define $B\subset \cup_{i=1}^2\Xi^+_i$ such such that: 
	\begin{itemize}
		\item $G$ is well defined on $(p_0, w_0)$ and $G(p_0, w_0)=(p_1, w_1)\in \cup_{i=1}^2\Xi^+_i$; 
		\item if $p_0\in V^{(i)}_\ve$, the transition point $\tilde p$ of the concatenation starting from $(p_0, w_0)$ is again in $V_\ve^{(i)}$. 
	\end{itemize}
	As in the reflective case, $G$ is clearly well defined on $B$ and the inverse mapping $G^{-1}$ is well defined on $G(B)$ by reversibility. Define then 
	\begin{equation}
		Y\eqdef \bigcap_{k\in\Z} G^{k}(B)
	\end{equation}
	to obtain a set of initial conditions which is invariant under $G$, where the first return map is infinitely-many  forward and backward well defined, and such that every trajectory starting with conditions in $Y$ realises a sequence in $\mathbb L$. It is then possible to construct the diagram \eqref{eq:diagram_refr}, and, following the same proof as in Proposition \ref{prop:conj_rifl}, proving the same result in the refractive case. 
		\begin{proposition}\label{prop:conj_refr}
		Given $P:Y\to \L$ defined as in \eqref{eq:proj}, it holds that \begin{enumerate}
			\item the diagram in \eqref{eq:diagram} commutes; 
			\item $P$ is bijective; 
			\item $P$ is continuous. 
		\end{enumerate}
	\end{proposition}
	This concludes the proof of Theorem \ref{thm:intro_chaos} in the refractive case as well. \\
	The chaoticity stated in Theorem \ref{thm:intro_chaos} could be related to the results of \cite{Delis20152448}, where the authors conclude, with a mixed numerical and analytical approach, that the three dimensional model where $V_I$ and $V_E$ are simply superimposed admits positive Lyapunov exponents, giving a first hint of chaos.\\
	Let us conclude by noticing that the geometric hypothesis on $D$ identified in Theorem \ref{thm:intro_chaos} is a sufficient condition to have local chaos at large inner energies: it is possible that, considering a different symbolic dynamics, such condition can be extended to other domains' shapes.

		\appendix
	\section{Analytical preliminaries}\label{sec:app}
	This Appendix is intended for the reader's convenience  as a complement to the paper, by giving more details on the analytical tools and results used in Sections \ref{sec:stability} and \ref{sec:chaos}. In particular, Section \ref{ssec:existence_fix_ends} summarises some known facts on the existence of fixed-ends solution for the outer and inner problem; in the case of a Keplerian potential we shall deal with the presence of the singularity at the origin, so it will be necessary to extend our notion of solution and, at the same time, to consider a suitable regularization procedure in three dimensions. In Section \ref{ssec:jacobi_der} we will introduce the concept of \textit{Jacobi distance} (see also \cite{ambrosetticotizelati}), specifying its properties, and of \textit{generating functions} $\S_{E\backslash I}$ that we use in both Section \ref{sec:stability} and \ref{sec:chaos}, including the proofs of Lemmas \ref{lem:der_SI_hom} and \ref{lem:der_SE_hom}. 
	
	\subsection{Existence of fixed-ends solution for the outer and inner problem}\label{ssec:existence_fix_ends}

	In this section we address the problem of finding fixed-ends solution for the outer or inner dynamics at given energies, which stays either outside or inside our domain $D$. By construction of our billiards, it is indeed crucial that, given two points on the domain's boundary $S=\partial D$, a solution of $z''=\nabla V_E(z)$ connecting them stays always outside D; on the other hand, and Keplerian arc connecting them can not exit from $D$. \\
	Let us start by stating the existence and uniqueness of the solution of the outer problem, which is ensured provided that the endpoints are sufficiently close to each others. 	
	 \begin{theorem}\label{thm:outer problem}
	 	For a fixed $\mathcal E>0$, there exists $\delta>0$ such that, for every $p_0, p_1\in S$ such that $|p_0-p_1|<\delta$, there exist a unique $T>0$ and a unique trajectory $z(\cdot; p_0, p_1): [0, T]\to \R^2$that is solution of the problem 
	 	\begin{equation}
	 		\begin{cases}
	 			z''(t)=-\omega^2 z(t), \quad &\forall t\in [0, T]\\
	 			\frac{1}{2}|z'(t)|^2-\frac{\omega^2}{2}|z(t)|^2=\mathcal E \quad &\forall t\in[0,T]\\
	 			z(0)=p_0, \ z(T)=p_1\\
	 			z(t)\notin \bar{D} &\forall t\in(0,T)
	 		\end{cases}
	 	\end{equation}
 	In particular: 
 	\begin{itemize}
 		\item if $p_0\neq p_1$, the arc belongs to the plane generated by $p_0$, $p_1$ and the origin; 
 		\item if $p_0=p_1$, the solution is a radial homothetic arc along the direction of the endpoints, which starts from $p_0$, is reflected back against the boundary of the Hill's region associated to the outer potential (that is, $\{z\in\R^3\ |\ |z|^2=2\mathcal E\}$), and returns on the same point (see Figure \ref{fig:homothetics}).
 	\end{itemize} 	
	 \end{theorem}
 	 	The proof of the above Theorem is straightforwardly based on a simple transversality argument starting from the homothetic arcs, which takes into account the star-convexity of $D$ with respect to the origin (for more details, see \cite{deblasiterracinirefraction}). \\
 	 	
 	 	As for the inner dynamics, we have the additional difficulty given by the presence of a singularity at the origin: to treat this case, it is necessary to introduce a more general definition of solution, called of \textit{collision-ejection}. 
 	 \begin{definition}
 	 	An arc $z: \R\to \R^3$ is called a \emph{collision-ejection solution} for the problem 
 	 	\begin{equation}\label{eq:inner_prob}
 	 		\begin{cases}
 	 		z''(t)=-\dfrac{\mu}{|z(t)|^3} z(t)\\
 	 		\dfrac12|z'(t)|^2-\dfrac{\mu}{|z(t)|}=\mathcal H
 	 		\end{cases}
 	 	\end{equation}
  		if there exists $t^*\in \R$ such that $z(t^*)=0$ and $z(t)$ is a $C^2$ solution of \eqref{eq:inner_prob} for any compact subset of $\R\setminus \{t^*\}$. 
 	 \end{definition}
 	 We highlight that collision-ejection solutions are a particular case of \emph{generalised solutions} as described in \cite{ambrosetticotizelati}. \\
 	Let us now fix two generic points $p_0,p_1\in\R^3\setminus \{0\}$, and consider the fixed-end Keplerian problem 
	\begin{equation}\label{eq:fixed_inner}
	\begin{cases}
		z''(t)=-\dfrac{\mu}{|z(t)|^3} z(t) \quad &t\in[0,T]\\
		\dfrac12|z'(t)|^2-\dfrac{\mu}{|z(t)|}=\mathcal H \quad &t\in [0, T]\\
		z(0)=p_0, \ z(T)=p_1
	\end{cases}	
	\end{equation}
 	 for some $T>0$. If the two points are not coincident nor collinear with the origin, it is always possible to find exactly two Keplerian arcs connecting them at a given energy (see also Figure \ref{fig:fix_arcs}).
 	  \begin{figure}
 	  	\centering
 	  	\includegraphics[width=0.35\linewidth]{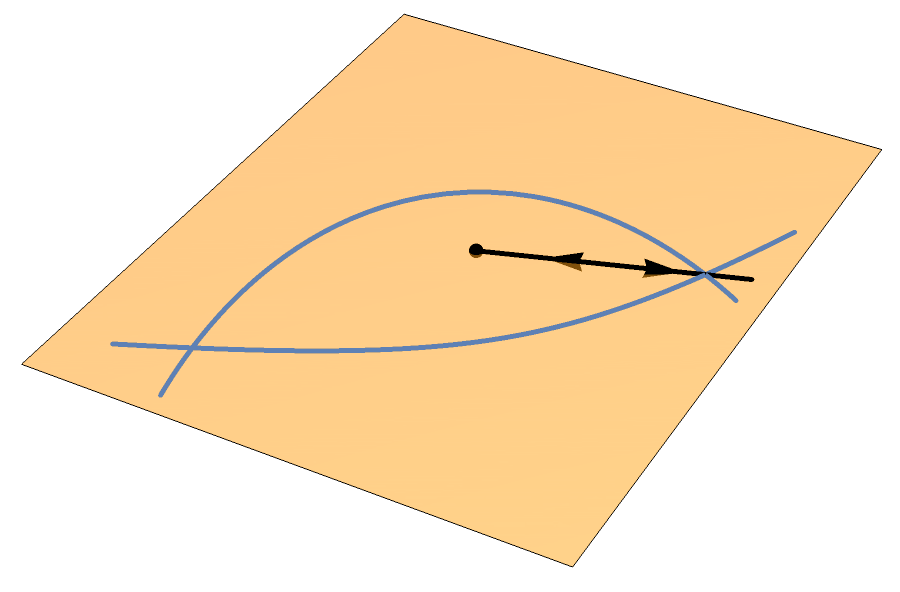}\qquad\qquad\qquad 
 	  	\includegraphics[width=0.35\linewidth]{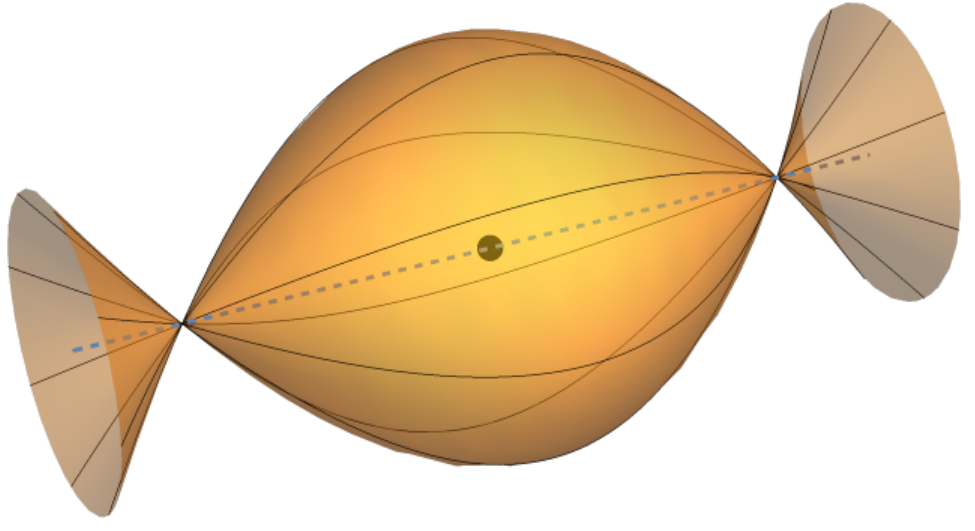}
 	  	\caption{Left: fixed-ends Keplerian outer arcs for the non-collisional and non-antipodal case (blue) and collision-ejection case (black). In the first case, the two arcs belong to the plane containing the origin and the two points. Right: Keplerian arcs connecting two antipodal points: they form a surface obtaining by rotating a Keplerian planar arc. }
 	  	\label{fig:fix_arcs}
 	  \end{figure}
 	\begin{theorem} \label{thm:inner problem 1}
 		Given $\H>0$ and two points $p_0,p_1\in \R^3\setminus\{0\}$, $p_0\neq p_1$ and not collinear with the origin, there are exactly two solutions of \eqref{eq:fixed_inner}, each of them belonging to the plane $\Pi$ generated by $p_0$, $p_1$ and the origin itself. \\
 		Moreover, one of these arcs is always homotopic to the segment connecting $p_0$ to $p_1$ in $\Pi$, while the other is not.  
 	\end{theorem}
 	The proof of the above result is completely analogous to the one in \cite[Lemma B.1]{deblasiterracinibarutelloChaotic}: once the plane $\Pi$ is defined, by the conservation of the angular momentum we already know that the arcs, if existent, belong to it; therefore, the dimensionality of the problem reduces to two and we can argue as in the planar case. \\
 	As a consequence of the above theorem, we have that, whenever the endpoints are not collinear with the origin and different from each other, it is always possible to recover the uniqueness of the Keplerian arc connecting them by simply fixing the homotopy class: for reason that will be clear later, we will always choose the arc belonging to the homotopy class which does not contain $\overline{p_0p_1}$, which we  call $z_I(\cdot; p_0, p_1)$. When the origin belongs to the segment between the endpoints, it is possible to find infinitely-many Keplerian arcs connecting them with the same energy $\H>0$: they are rotations of the same planar arc (see Figure \ref{fig:fix_arcs}). \\	
 	We stress that, in both the inner and outer case, any solution of the problem is planar as a consequence of the conservation of the angular momentum, depending on the radial symmetry of $V_E$ and $V_I$.\\
 	While in the case described in Theorem \ref{thm:inner problem 1} we always find classical $C^2$ solutions of the associated problem, this is not the case when $p_0=p_1$: indeed we can not exclude that the systems enters into collision. For this reason, a \textit{regularization} of three-dimensional Kepler problem is in order: in the present paper, we propose the \textit{Kustanheimo-Stiefel} regularization in the notation of Waldvogel (see \cite{waldvogel2008quaternions}), which can be considered as a three dimensional generalization of the classical Levi-Civita one (see \cite{Levi-Civita}). Waldvogel's construction involves the definition of a suitable four-dimensional space, called \textit{quaternion space}, and of a conformal transformation  that, working like the classical complex square in Levi-Civita case, allows to conjugate Problem \eqref{eq:fixed_inner} to a harmonic non-singular one. 
 	Here we propose a brief summary of Waldvogel's argument: the interested reader can find more details in \cite{waldvogel2008quaternions}.\\
 	Let us start by defining three imaginary units $\hi, \hj, \hk$ satisfying
 	\begin{equation}
 		\hi^2=\hj^2=\hk^2=-1, \quad \hi \hj=-\hj \hi=\hk, \quad \hj \hk=-\hk \hj=\hi, \quad \hk \hi=-\hi \hk =\hj:   
 	\end{equation}
 	the \textit{quaternion space} $\U$ is given by 
 	\begin{equation}
 		\U \eqdef \left\{\un{x}=\x{0}+\x{1}\hi+\x{2}\hj+\x{3}\hk \ |\ \x{n}\in \R\ \forall n=0,\dots,3\right\}\simeq \R^4. 
 	\end{equation}
 	The space $\U$ is a non-commutative algebra on $\R$, and it is immediate to notice that 
 	\begin{equation}
 		\R^3\simeq\left\{\un{x}\in \U\ |	\ \x{3}=0\right\}: 
 	\end{equation}
 	with an abuse of notation, we will identify any vector in $\R^3$ with elements in $\U$ whose last component is null. It is possible to define the usual norm on $\U$, given by
 	\begin{equation}
 		|\un{x}|\eqdef\sum_{i=0}^3\left(\x{i}\right)^2, 
 	\end{equation} 
	the adjoint operator  
	\begin{equation}
		\un{x}^*\eqdef \x{0}+\x{1}\hi+\x{2}\hj-\x{3}\hk, 
	\end{equation}
 and the mapping 
 \begin{equation}
 	KS: \U \to \R^3 \quad \un{x}\mapsto \un{z}=\un{x}\ \un{x}^*. 
 \end{equation}
The mapping $KS$ corresponds to the \textit{Kustanheimo-Stiefel mapping}, and it is easy to verify that, for any $\un{z}\in \R^3$, its preimages through $KS$ are given by 
\begin{equation}\label{eq:KS_preim}
	\begin{aligned}
		\un{x}=\un{x}_0\left(\cos\phi + \hk \sin{\phi}\right), \quad \phi\in\mathbb T,\quad \un{x}_0=\frac{\un{z}+|\un{z}|}{\sqrt{2(z^{(0)})+|\un{z}|}}\in \R^3. 
	\end{aligned}
\end{equation}
 The map $KS$ is a generalization of the classical Levi-Civita transformation in four dimensions: to mimic the regularization algorithm holding in the planar case, it is necessary to understand how it behaves under differentiation. By imposing the bilinear relation 
 \begin{equation}\label{eq:bilinear}
 		2(x^{(3)}dx^{(0)}-x^{(2)}dx^{(1)}+x^{(1)}dx^{(2)}-x^{(0)}dx^{(3)})=0, 
 \end{equation}
one has straightforwardly that $d\left(KS\right)(\underline x)=2\un{x}\ \un{x}^*$: when $\un{x\in\R^3}$, then $\un{x}=\un{x}^*$ and  $d(KS)(\un{x})=2\un{x}d\un{x}$, which is exactly the derivation rule of the complex square. \\
By using the quaternions formalism and the relations described above, it is possible to construct a change in both spatial coordinates (through $KS$) and time parameter, described in full details in \cite{waldvogel2008quaternions}, which leads to the following result. 
\begin{proposition}\label{prop:KS_conj}
		For fixed $p_0, p_1\in\R^3$ and $\H>0$, problem \eqref{eq:fixed_inner} is conjugated, through the transformations
	\begin{equation}
		\frac{d}{dt}=\frac{1}{2 r}\frac{d}{ds}, \quad \underline x\ \underline x^*=\underline z, \quad r=|\underline z|=|\underline x|^2, 
	\end{equation}
	to the regularised system, in  $\mathbb U$, 
	\begin{equation}
		 \begin{cases}
			\ddot{\underline x}(s)-\Omega^2 \underline x(s)=0, \quad &s\in [-\tilde T, \tilde T]\\
			\frac{1}{2}|\dot{\underline{x}}(s)|^2+\frac{\Omega^2}{2}|\underline{x}(s)|^2=\H_{KS}, &s\in[-\tilde T,\tilde T]\\
			\underline x(-\tilde T)=\underline x_0, \ \underline x(\tilde T)=\underline x_1, 
		\end{cases}
	\end{equation}
	with $\underline x(s)\in\mathbb U$ for every $s\in[-\tilde T,\tilde T]$, for a suitable $\tilde T>0$, $\Omega^2=2 \H$, $\H_{KS}=\mu$, and endpoints 
	\begin{equation}
		\underline{x}_0=-\frac{\underline p_0+|\underline{p_0}|}{\sqrt{2\left(p_0^{(0)}+|\underline p_0|\right)}}, \quad 
		\underline{x}_1=\frac{\underline p_1+|\underline{p_1}|}{\sqrt{2\left(p_1^{(0)}+|\underline p_1|\right)}}. 
	\end{equation}
\end{proposition}
Let us highlight that, by virtue of Eq. \eqref{eq:KS_preim}, we have one degree of freedom in choosing the preimages $\un x_0, \un x_1$ of $p_0, p_1$. The choice made in Proposition  \ref{prop:KS_conj} has two main reasons: first of all, $\un x_0,\un x_1\in \R^3$; secondly, this is the same choice made in the Levi-Civita case when searching for the collision trajectory (see for example \cite{deblasiterracinibarutelloChaotic}), where one takes respectively the negative and positive determination of the complex square root. \\
Using the above proposition, it is possible to treat the case $p_0=p_1$. 
\begin{theorem}\label{thm:inner problem 2}
Let $\H>0$ and take $p_0\in \R^3\setminus \{0\}$. Then there exists $T>0$ such that system 
\begin{equation}
	\begin{cases}
		z''(t)=-\dfrac{\mu}{|z(t)|^3} z(t) \quad &t\in[0,T]\\
		\dfrac12|z'(t)|^2-\dfrac{\mu}{|z(t)|}=\mathcal H \quad &t\in [0, T]\\
		z(0)=z(T)=p_0
	\end{cases}	
\end{equation}
admits a unique collision-ejection solution,  parallel to the radial direction corresponding to $p_0$, which after collision is reflected back along the same line. 	
\end{theorem}

\begin{proof}\label{eq:fixed_KS}
	Without loss of generality, let us suppose that $p_0=(R, 0,0)$, $R>0$: by central symmetry, the following reasoning applies to any $p_0\in \R^3\setminus\{0\}$. According to Proposition \ref{prop:KS_conj}, the problem is conjugated in $\U$ to
	\begin{equation}
		\begin{cases}
			\ddot{\underline x}(s)+\Omega^2\underline x(s)=0\\
			\frac12|\dot{\underline x}(s)|^2+\frac{\Omega^2}{2}|\underline x(s)|^2=\H_{KS}\\
			\underline x(-\tilde T)=\underline x_0=-\underline x_1=\underline x(\tilde T)=-\sqrt{R}\in\mathbb U
		\end{cases}
	\end{equation}
whose solution is trivially given by 
\begin{equation}
	\tilde{\underline x}(s)=\frac{\sqrt{2 \H_{KS}}}{\Omega}\sinh(\Omega s)\in \mathbb U, \quad s\in[-\tilde T,\tilde T], \quad \tilde T=\frac{1}{\Omega}arcsinh\left(\frac{\Omega\sqrt{R}}{\sqrt{2\H_{KS}}}\right). 
\end{equation}
Returning to the physical space, the above solution corresponds to a collision-ejection arc given by 
\begin{equation}
	\tilde z(t)=\tilde{\underline x}(s(t))\ \tilde{\underline x}^*(s(t))=\tilde{\underline x}^2(s(t)), \quad t\in[0, T_I]
\end{equation}
for a suitable $T_I>0$. The solution $\tilde z(t)$ is always parallel to the $x-$axis, starts from $p_0$, bounces at the origin and returns at the initial point along the same direction.
\end{proof}

There is a strong relation between classical solutions of Problem \eqref{eq:fixed_inner} and collision ones: it is in fact possible to prove that, taking the unique solution of the problem which is not homotopic to $\overline{p_0p_1}$ and taking the limit $p_0\to p_1$, the arc $z_I(\cdot; p_0, p_1)$ will tend precisely to the straight collision-ejection solution in the direction of $p_1$. This is the reason for the choice if the homotopy class made before.  \\
Now that the existence of an arc connecting points in $\R^3$ is ensured, let us return to our original problem: take $p_0,p_1\in S$ and search for a Keplerian inner arc connecting them. If the points are not antipodal to each others, we are able to find a solution of the inner problem that goes from $p_0$ to $p_1$; nevertheless, we still have to prove that such solution is always internal to $D$. This is true whenever the inner energy is sufficiently large, as a consequence of the following classical result in Celestial Mechanics (cfr. \cite[pp.273-274]{battin}). 
\begin{lemma}\label{lem:battin}
	Let us take $p_0, p_1\in \R^3\setminus\{0\}$ not aligned with the origin, and consider the arcs $z_I(\cdot; p_0, p_1)$ not homotopic to $\overline{p_0p_1}$ for increasing values of $\H$. Then, it holds that they converge in $C^0-$norm to the broken line composed by the segments $\overline{p_00}$ and $\overline{0p_1}$.  
\end{lemma}
Putting together all the informations above, we can finally state the final result of this Section, claiming the existence and uniqueness of inner arc provided that the energy is sufficiently large. 
\begin{theorem}\label{thm:existence_inner}
	There exists $\overline\H>0$ such that, for every $\H>\overline \H$, $p_0, p_1\in S$ not antipodal, there is a unique solution of the inner fixed-ends problem \eqref{eq:fixed_inner}, denoted by $z_I(\cdot; p_0, p_1)$, which is not homotopic to $\overline{p_0p_1}$ and stays completely inside $D$. If $p_0=p_1$, $z_I(\cdot; p_0,p_1)$ is a collision-ejection arc.
\end{theorem}
The proof is a simple consequence of Theorems \ref{thm:inner problem 1}, \ref{thm:inner problem 2} and Lemma \ref{lem:battin}, along with the star-convexity of the domain $D$ and the compactness of its boundary $S$. 

\subsection{Jacobi distance, generating function and their derivatives}	\label{ssec:jacobi_der}
 A very useful analytical tool to describe the geometric properties of outer and inner arcs in our billiard models, as well as to construct suitable first return maps, is given by the \emph{Jacobi length}. In the present section we summarise the principal definitions and properties we use in the whole paper, along with the proofs of Lemmas \ref{lem:der_SE_hom} and \ref{lem:der_SI_hom}; for more details, one can see \cite{deblasiterracinirefraction} and the references therein.\\
 	
Let us start by considering a regular potential $V:\R^3\to \R$, and, given $p_0, p_1\in \R^3$, let us consider $z\in H^1([0,T],\R^3)$ such that $z(0)=p_0$ and $z(T)=p_1$: given $E\in \R$, the \textit{Jacobi length} of $z$ is given by 
\begin{equation}
	\mathcal L(z(\cdot))=\int_0^T|z'(t)|\sqrt{E+V(z(t))}\ dt. 
\end{equation}   
It is possible to show (see for example \cite{deblasiterracinirefraction}) that solutions of the fixed-ends problem at energy $E$ defined by
\begin{equation}\label{eq:fix_gen}
	\begin{cases}
		z''(t)=\nabla V(z(t))\\
		\frac{1}{2}|z'(t)|^2-V(z(t))=E\\
		z(0)=p_0, \ z(T)=p_1
	\end{cases}
\end{equation} 
correspond to critical points of the corresponding Jacobi length in the space of the $H^1$ functions which connect $p_0$ to $p_1$; moreover, the functional $\mathcal L$ is invariant under reparametrizations. \\
Let us now suppose that, given $p_0,p_1$ as endpoints, the problem \eqref{eq:fix_gen} admits a unique solution, denoted by $z(\cdot; p_0, p_1)$: it is then possible to use the functional $\mathcal L$ to define the \emph{Jacobi distance} between $p_0$ and $p_1$ as $d(p_0, p_1)=\mathcal L\left(z(\cdot; p_0, p_1)\right)$. The function $d(\cdot, \cdot)$ is well defined for every pair of endpoints for which the solution $z$ is unique, and depends on the potential $V$, as well as on the energy $E$. Whenever well defined, $d(\cdot, \cdot)$ is of class $C^\infty$ in both variables, and it holds 
\begin{equation}\label{eq:derL}
	\nabla_1d(p_0, p_1)=-\sqrt{V(p_0)}\dfrac{z'(0)}{|z'(0)|}, \quad \nabla_2d(p_0, p_1)=-\sqrt{V(p_1)}\dfrac{z'(T)}{|z'(T)|}, 
\end{equation}
where $\nabla_{1\backslash 2}$ refer to the gradient respectively with respect to the first and second point. \\
Returning to the billiard model, we can apply the above definitions to the inner and outer dynamics to characterise arcs which start and end on the domain's boundary $S$. \\
Let us start with the reflective case, where the potential involved is the sole inner one: in principle, to define properly the Jacobi distance of two points on $S$ connected by an inner arc we have two problems: the lack of uniqueness of Keplerian hyperbol\ae \  connecting two non-antipodal points and the singularity of $V_I$ at the origin. The first problem is solved by choosing suitably the homotopy class where searching for our arcs, as it is done in Theorem \ref{thm:existence_inner}. As for the singularity, we already explained in \ref{prop:KS_conj} how it is possible to conjugate the Kepler problem to a  regular fixed ends problem in the quaternion space $\U$: such conjugation has consequences on the definition of the Jacobi length in the inner case, which we will denote with $\mathcal L_I$. In particular, by direct computations one can prove the following lemma. 
\begin{lemma}\label{lem:L_I_to_KS}
	Let $p_0, p_1\in \R^3$ and let $z_I(\cdot; p_0, p_1)$ be the solution (possibly of collision-ejection) of \eqref{eq:fixed_inner} of Theorem \ref{thm:existence_inner}, while $\un x(\cdot; \un{x}_0, \un{x}_1)$ is the corresponding conjugated solution in \eqref{eq:fixed_KS}. Then it holds 
	\begin{equation}
		\mathcal L_I(z_I(\cdot; p_0, p_1))=2\mathcal L_{KS}\left(\un x(\cdot; \un{x}_0, \un{x}_1)\right)\eqdef \int_{-\tilde T}^{\tilde T}|\dot{ \un{x}}(\tau)|\sqrt{\dfrac{\Omega^2}{2}|\un{x}(\tau)|^2+\H_{KS}}\ d\tau. 
	\end{equation}
In particular, the formula holds when $p_0=p_1$, then in the case of collision-ejection arcs. 
\end{lemma} 
Through the above lemma, we are able to define the \textit{inner Jacobi distance} $d_I(p_0, p_1)\eqdef\mathcal L_I(z_I(\cdot; p_0, p_1))$ for every pair of non antipodal\footnote{Let us notice that the non antipodality is not a restrictive hypothesis for our analysis, since we will never work with antipodal endpoints.  } points, even coincident.\\
Let us now describe how to use the Jacobi distance to give a variational interpretation of the reflection law. To this aim, let us take two points $p_0, p_1\in S$, and suppose that, taking $\H$ sufficiently large, they are connected by a concatenation of inner arcs passing through $\tilde p\in S$: considering $d_I(p_0, \tilde p)$ and $d_I(\tilde p, p_1)$, it is possible to define the total Jacobi length of the concatenation through $\tilde p$ as   
\begin{equation}
	f: \tilde p \mapsto d_I(p_0, \tilde p)  + d_I(\tilde p, p_1). 
\end{equation}
The point $\tilde p$ is critical for $f$ constrained to $S$ if and only if for every $e\in T_{\tilde p}S$ it holds $\scp{\nabla f(\tilde p), e }=0$, namely, recalling \eqref{eq:derL}\footnote{With an abuse of notation, here and in the following we will denote with $T>0$ the final time necessary to arrive to the second endpoint, which clearly depends on the points and the potential. }, 
\begin{equation}
	\scp{z_I'(T; p_0, \tilde p), e}=\scp{z'_I(0; \tilde p, p_1), e}, 
\end{equation}
which defines precisely the reflection law. In other words, \emph{asking that a concatenation is connected by a reflection is equivalent to ask that the transition point is a critical point for the total Jacobi length of two consecutive arcs.}\\
Let us now pass in local coordinates, and, given $p_0, p_1\in S$ not antipodal and $\H$ sufficiently large, let us consider two local charts $\Gamma^{(i)}: U_i\to S$, $i=0,1$, around the two points. We can then rewrite locally the Jacobi distance through the function 
\begin{equation}
	\S_I: U_0\times U_1\to \R, \quad (u_0, v_0, u_1, v_1)\mapsto d_I\left(\Gamma^{(0)}(u_0,v_0), \Gamma^{(1)}(u_1,v_1)\right). 
\end{equation}
The function $\S_I$ is at least $C^2$, and by the chain rule one finds the derivatives' formulas listed in Eqs. \eqref{eq:der1S}. Such derivatives depend on the parameterisation and find a closed explicit formula only knowing the initial or final velocity vector of the inner arc. In our computations, we will need to find explicit formulas only in the case of homothetic arcs, in the proof of 	Lemma \ref{lem:der_SI_hom}. \\

As for the refractive case, once one has $p_0,p_1\in S$ sufficiently close, by theorem \ref{thm:outer problem} there exists a unique outer arc $z_E(\cdot; p_0, p_1)$ connecting them, and using the same reasoning as before one can define a \emph{Jacobi distance} $d_E(p_0, p_1)\eqdef \mathcal L_E(z_E(\cdot; p_0, p_1))$, where $\mathcal L_E$ is the Jacobi length relative to $V_E$, and the corresponding representation in local coordinates $\S_E(u_0, v_0, u_1, v_1)$. Here, the outer potential $V_E$ is smooth, and then no regularisation procedures are necessary. Of course, the inner dynamics is described by the same equations as in the reflective model, and by analogy with the previous case, whenever $p_0, p_1\in S$ are connected by a concatenation of refracted arcs passing through a point $\tilde p\in S$, it results that the latter is a critical point for the function 
\begin{equation}
	p\mapsto d_E(p_0, p)+d_I(p, p_1): 
\end{equation}
in local coordinates, one retrieves Eqs. \eqref{eq:refrS}. Let us now take into account Lemmas \ref{lem:der_SE_hom}, \ref{lem:der_SI_hom}, used in Section \ref{sec:stability} to compute the linear stability of homothetic trajectories. Taking $\bar p\in S$ central configuration according to Definition \ref{def:cc}, consider a local chart $\Gamma$ as in Remark \ref{rem:par}.
Lemmas \ref{lem:der_SE_hom} and \ref{lem:der_SI_hom} provide respectively the expressions of the second derivatives of $\S_E$ and $\S_I$ when computed in correspondence of a homothetic arc, namely, in the case $u_0=u_1=\bu$ and $v_0=v_1=\bv$, as functions of the geometric parameters and the geometric properties of $S$ in $\bp$. \\
As for the outer dynamics, the proof is completely analogous to the one presented, in the case of a two-dimensional dynamics, in \cite[Section 5.2]{deblasiterracinirefraction}: we will focus then on the proof the Lemma in the inner case, where the use of a different regularization technique requires some adjustments to what has been already done in two dimensions.  
\begin{proof}[Proof of Lemma \ref{lem:der_SI_hom}]
	For the sake of simplicity, let us suppose that $\bp$ belongs to the $x-$axis, and in particular $\bp=(R, 0,0)$: in any other case, a simple rotation of the reference frame will lead to the same conclusion. Let us then consider the collision-ejection inner solution $z_I(t; \bp, \bp)$, and the corresponding solution $\overline{\un x}(\tau)$ in the regularised system, which, according to Proposition \ref{prop:KS_conj}, has endpoints $\un x_0=-\sqrt{R}\in \U$ and $\un x_1=\sqrt{R}\in \U$: from Lemma \ref{lem:L_I_to_KS}, it is sufficient to find the second derivatives of $\mathcal L_{KS}$ to find the ones of $\mathcal L_I$.
	By the choice of the parameterisation $\Gamma,$ we have that $\bp = |\bp|\hat x$, $\Gamma_u(\bu, \bv)=\hat z$, $\Gamma_v(\bu, \bv)=\hat y$,  $ N(\bu, \bv)=-\hat x$, where $\hat x, \hat y$ and $\hat z$ are the canonical basis of $\R^3\subset \U$.\\
	Since we need to consider infinitesimal variations on $S$ from $\bar p$, it is necessary to compute the local preimages of $S$ through the transformation $KS$ locally around $\un x_0$ and $\un x_1$. Identifying, with an abuse of notation, $\Gamma(u,v)\in \U$, let us consider the subsets of $\U$ parameterised by
	\begin{equation}
		\Gamma^{\pm}: U\to \mathbb U, \ \Gamma^{-}(u,v)=-\frac{ \Gamma (u,v)+|{\Gamma(u,v)}|}{\sqrt{2\left(\Gamma^{(0)}(u,v)+| \Gamma(u,v)|\right)}}, \ \Gamma^{+}(u,v)=\frac{ \Gamma (u,v)+|{\Gamma(u,v)}|}{\sqrt{2\left(\Gamma^{(0)}(u,v)+| \Gamma(u,v)|\right)}}. 
	\end{equation}
	By construction, for every $(u, v)\in U$ it holds $\Gamma^{\pm}(u,v)\in \R^3$; moreover, $\Gamma^{-}(\bu, \bv)=\un x_0$ and $\Gamma^{+}(\bu, \bv)=\un x_1$. 
	Let us now compute the tangent vectors of $\Gamma^{\pm}$, by taking into account the bilinear relation \eqref{eq:bilinear}:
	\begin{equation}
		\left(\Gamma^-(u,v)\right)^2= \Gamma(u,v)\ \Rightarrow\  2\Gamma^-(\bu,\bv)\Gamma^-_u(\bu,\bv)=\Gamma_u(\bu,\bv):
	\end{equation}
	considering $\Gamma^-_u(\bu, \bv)=x^{(0)}+x^{(1)}\hi+x^{(2)}\hj+x^{(3)}\hk$ as an unknown element of $\U$, one has
	\begin{equation}\label{eq:gammaminus}
		-2\sqrt{R}\left(x^{(0)}+x^{(1)}\hi+x^{(2)}\hj+x^{(3)}\hk\right)=\hj \Rightarrow \Gamma_u^-(\bu, \bv)=-\frac{1}{2\sqrt{R}}\hj; 
	\end{equation}
	with the same reasoning, one can find $\Gamma^+_u(\bu,\bv)=(2\sqrt{R})^{-1}\hj$, $\Gamma^-_v(\bu,\bv)=-(2\sqrt{R})^{-1}\hi$ and $\Gamma^+_u(\bu,\bv)=(2\sqrt{R})^{-1}\hi$. \\
	Let us now pass to the second derivatives of $\Gamma^\pm$: differentiating again Eq. \eqref{eq:gammaminus}, one has
	\begin{equation}
		2\left(\Gamma_u^-(u,v)\right)^2+2\Gamma^-(u,v)\Gamma_{uu}^-(u,v)=\Gamma_{uu}(u,v).  
	\end{equation}
    Projecting along the first component in $\U$, and again treating $\Gamma^-_{uu}(\bu,\bv)=x^{(0)}+x^{(1)}\hi+x^{(2)}\hj+x^{(3)}\hk$ as an unknown vector, one can take into account Eqs. \eqref{eq:par_giusta}, finding
	\begin{equation}
		\begin{aligned}
			&\left(2\left(\Gamma_u^-(\bu,\bv)\right)^2+2\Gamma^-(\bu,\bv)\Gamma_{uu}^-(\bu,\bv)\right)^{(0)}=\Gamma_{uu}^{(0)}(\bu,\bv) \\
			&\Rightarrow \left(2\left(-\frac{1}{2\sqrt{R}\hj}^2-2\sqrt{R}\left(x^{(0)}+x^{(1)}\hi+x^{(2)}\hj+x^{(3)}\hk\right)\right)\right)^{(0)}=-\kappa_1\\
			&\Rightarrow \left(\Gamma_{uu}^-(\bu,\bv)\right)^{(0)}=-\frac{1}{2\sqrt{R}}\left(\frac{1}{2R}-\kappa_1\right), 
		\end{aligned}
	\end{equation}
	and, with the same reasoning, 
	\begin{equation}
		\begin{aligned}
			&\left(\Gamma_{uu}^+(\bu,\bv)\right)^{(0)}=\frac{1}{2\sqrt{R}}\left(\frac{1}{2R}-\kappa_1\right), \ \left(\Gamma_{vv}^-(\bu,\bv)\right)^{(0)}=-\frac{1}{2\sqrt{R}}\left(\frac{1}{2R}-\kappa_2\right), \\ &\left(\Gamma_{vv}^+(\bu,\bv)\right)^{(0)}=\frac{1}{2\sqrt{R}}\left(\frac{1}{2R}-\kappa_2\right)\\
			& \left(\Gamma_{uv}^-(\bu,\bv)\right)^{(0)}=\left(\Gamma_{vu}^-(\bu,\bv)\right)^{(0)}= \left(\Gamma_{uv}^+(\bu,\bv)\right)^{(0)}=\left(\Gamma_{vu}^+(\bu,\bv)\right)^{(0)}=0. 
		\end{aligned}
	\end{equation}
Let us now define a new Jacobi distance in $\U$, given by 
	\begin{equation}
		\S_{KS}(u_0,v_0,u_1,v_1)=d_{KS}(\Gamma^-(u_0, v_0), \Gamma^+(u_0,v_0))=\mathcal L_{KS}(\underline x(\cdot; \Gamma^-(u_0, v_0), \Gamma^+(u_0,v_0)): 
	\end{equation}
	from Lemma \ref{lem:L_I_to_KS}, it holds $\S_I(u_0,v_0,u_1,v_1)=2 \S_{KS}(u_0,v_0, u_1, v_1)$, and the derivatives of $\S_{KS}$ are in the same form already found for $\S_{E\backslash I}$ in Eqs. \eqref{eq:der1S}. \\
	Let us now compute the homogeneous derivative $\partial_{1,u}^2S_{KS}(\bu,\bv, \bu, \bv)$; all the other derivatives can be found following the same reasoning: 
	\begin{equation}
		\partial_{1,u}\S_{KS}(u_0, v_0, u_1, v_1)=-\frac{1}{\sqrt{2}}\scp{\dot{\un x}\left(-\tilde T; \Gamma^-(u_0, v_0), \Gamma^+(u_1, v_1)\right), \Gamma^-_u(u_0, u_1)}.
	\end{equation}
	Let us denote with $e_u^-\eqdef \Gamma^-(u_0, v_0)/|\Gamma^-(u_0, v_0)|$: deriving the above equation in $u_0$, one finds 
	\begin{equation}\label{eq:der2uu}
		\begin{aligned}
		\partial_{1,u}^2\S_{KS}(u_0, v_0, u_1, v_1)=&-\frac{1}{\sqrt{2}}|\Gamma^-(u_0, v_0)|^2\scp{\frac{d}{d\tau}\left(\partial_{e_u^-}\un x \left(\tau; \Gamma^-(u_0, v_0), \Gamma^+(u_1, v_1) \right)\right)_{|\tau=-\tilde T}, \Gamma^-_u(u_0, u_1)}\\
		& -\frac{1}{\sqrt{2}} \scp{\dot{\un x}(-\tilde T; \Gamma^-(u_0, v_0), \Gamma^+(u_1, v_1)), \Gamma^-_{uu}(u_0, u_1)}, 
		\end{aligned}
	\end{equation}
	where $\partial_{e_u^-}\un x \left(\tau; ; \Gamma^-(u_0, v_0), \Gamma^+(u_1, v_1) \right)$ is the variation of the solution of the fixed-ends regularised problem moving the first point in the direction of $e_u^-$.  By standard variational arguments and direct computations (see \cite{deblasiterracinirefraction} for more details), one can find, in the homothetic, 
	\begin{equation}\label{eq:variazioni}
		\begin{aligned}
			& \underline{\bar x}(-\tilde T)=-2\sqrt{R}\in \mathbb U, \quad \underline{\bar x}(\tilde T)=2\sqrt{R}\in \mathbb U,
			\  \underline{\dot{\bar x}}(-\tilde T)=\underline{\dot{\bar x}}(\tilde T)= \sqrt{2 \H_{KS}+\Omega^2R}\in \mathbb U\\
			&\frac{d}{d\tau}\left(\partial_{e_u^-}\underline{{\bar x}}(\tau)\right)_{|_{\tau=-\tilde T}}=-\frac{\H_{KS}+\Omega^2R}{\sqrt{R}\sqrt{2 \H_{KS}+\Omega^2R}}\Gamma_u^-(\bu,\bv), \\ 
			&\frac{d}{d\tau}\left(\partial_{e_v^-}\underline{{\bar x}}(\tau)\right)_{|_{\tau=-\tilde T}}=-\frac{\H_{KS}+\Omega^2R}{\sqrt{R}\sqrt{2 \H_{KS}+\Omega^2R}}\Gamma_v^-(\bu,\bv)\\
			&\frac{d}{d\tau}\left(\partial_{e_u^+}\underline{{\bar x}}(\tau)\right)_{|_{\tau=-\tilde T}}=-\frac{\H_{KS}}{\sqrt{R}\sqrt{2 \H_{KS}+\Omega^2R}}\Gamma_u^+(\bu,\bv)\\
			&\frac{d}{d\tau}\left(\partial_{e_v^+}\underline{{\bar x}}(\tau)\right)_{|_{\tau=-\tilde T}}=\frac{\H_{KS}}{\sqrt{R}\sqrt{2 \H_{KS}+\Omega^2R}}\Gamma_v^+(\bu,\bv)\\
			&\frac{d}{d\tau}\left(\partial_{e_u^-}\underline{{\bar x}}(\tau)\right)_{|_{\tau =\tilde T}}=-\frac{\H_{KS}}{\sqrt{R}\sqrt{2 \H_{KS}+\Omega^2R}}\Gamma_u^-(\bu,\bv)\\
			&\frac{d}{d\tau}\left(\partial_{e_v^-}\underline{{\bar x}}(\tau)\right)_{|_{\tau=\tilde T}}=-\frac{\H_{KS}}{\sqrt{R}\sqrt{2 \H_{KS}+\Omega^2R}}\Gamma_v^-(\bu,\bv)\\
			&\frac{d}{d\tau}\left(\partial_{e_u^+}\underline{{\bar x}}(\tau)\right)_{|_{\tau=\tilde T}}=\frac{\H_{KS}+\Omega^2R}{\sqrt{R}\sqrt{2 \H_{KS}+\Omega^2R}}\Gamma_u^+(\bu,\bv)\\
			&\frac{d}{d\tau}\left(\partial_{e_v^+}\underline{{\bar x}}(\tau)\right)_{|_{\tau=\tilde T}}=\frac{\H_{KS}+\Omega^2R}{\sqrt{R}\sqrt{2 \H_{KS}+\Omega^2R}}\Gamma_v^+(\bu,\bv), 
		\end{aligned}
	\end{equation}  
	where $e_v^\pm$, $e_u^+$ are defined analogously to $e_u^-$. Inserting \eqref{eq:variazioni} into \eqref{eq:der2uu}, and computing the other second derivatives, one finally obtains Eqs. \eqref{eq:der2S}. 
\end{proof}


\end{document}